\documentstyle[extras-alg,11pt]{article}

\newcommand{\qq}{\mbox{\boldmath $q$}}

\newcommand{\xx}{\mbox{\boldmath $x$}}
\newcommand{\uu}{\mbox{\boldmath $u$}}
\newcommand{\vv}{\mbox{\boldmath $v$}}
\newcommand{\pp}{\mbox{\boldmath $p$}}
\newcommand{\cc}{\mbox{\boldmath $c$}}

\newcommand{\mm}{\mbox{\boldmath $m$}}

\newcommand{\ga}{\mbox{$\gamma$}}
\newcommand{\al}{\mbox{$\alpha$}}
\newcommand{\bt}{\mbox{$\beta$}}
\newcommand{\dt}{\mbox{$\delta$}}
\newcommand{\th}{\mbox{$\theta$}}
\newcommand{\Ga}{\mbox{$\Gamma$}}
\newcommand{\ltwo}{\mbox{$l_2$}}
\newcommand{\lone}{\mbox{$l_1$}}

\newcommand{\argmax}{\mbox{\rm argmax}}

\newcommand{\la}{\leftarrow}

\newcommand{\minimize}{\mbox{\rm minimize }}
\newcommand{\maximize}{\mbox{\rm maximize }}

\newcommand{\st}{\mbox{\rm subject to }}
\newcommand{\CN}{\mbox{${\cal N}$}}
\newcommand{\CS}{\mbox{${\cal S}$}}
\newcommand{\CM}{\mbox{${\cal M}$}}

\newcommand{\Z}{\mbox{\rm\bf Z}}
\newcommand{\Zplus}{\Z^+}
\newcommand{\R}{\mbox{\rm\bf R}}
\newcommand{\Rplus}{\R_+}

\newcommand{\implies}{\mbox{${\Rightarrow}$}}

\makeatletter
\def\fnum@figure{{\bf Figure \thefigure}}
\def\fnum@table{{\bf Table \thetable}}

\long\def\@mycaption#1[#2]#3{\addcontentsline{\csname
 ext@#1\endcsname}{#1}{\protect\numberline{\csname
  the#1\endcsname}{\ignorespaces #2}}\par
     \begingroup
       \@parboxrestore
          \small
       \@makecaption{\csname fnum@#1\endcsname}{\ignorespaces
#3\endgroup}
      }

\newcommand{\RNB}{{\bf ADNB}}

\begin{document}

\title{Rational Convex Programs, Their Feasibility,\\
and the Arrow-Debreu Nash Bargaining Game}

\author{
{\Large\em Vijay V. Vazirani}\thanks{College of Computing,
Georgia Institute of Technology, Atlanta, GA 30332--0280,
E-mail: {\sf vazirani@cc.gatech.edu}.
}
}

\date{}
\maketitle

\begin{abstract}
Over the last decade, combinatorial algorithms have been obtained for exactly solving several nonlinear convex programs. We first provide a 
formal context to this activity by introducing the notion of {\em rational convex programs} -- this also enables us to identify a
number of questions for further study. 
So far, such algorithms were obtained for total problems only. Our main contribution is developing the methodology for handling non-total
problems, i.e., their associated convex programs may be infeasible for certain settings of the parameters. 

The specific problem we study pertains to a Nash bargaining game, called \RNB, which is derived from the linear case of the Arrow-Debreu 
market model. We reduce this game to computing an equilibrium in a new market model called {\em flexible budget market}, and we obtain 
primal-dual algorithms for determining feasibility, as well as giving a proof of infeasibility and finding an equilibrium.
We give an application of our combinatorial algorithm for \RNB \ to an important ``fair'' throughput allocation problem on a wireless channel.
\end{abstract}

\bigskip
\bigskip
\bigskip
\bigskip
\bigskip
\bigskip
\bigskip
\bigskip
\bigskip
\bigskip
\bigskip
\bigskip
\bigskip
\bigskip
\bigskip
\bigskip
\bigskip
\bigskip

\pagebreak

\section{Introduction}
\label{sec.intro}

The fascinating question of computability of market equilibria, which has been studied extensively within TCS over the last decade, 
has provided the area of algorithms a new direction, namely the design of efficient combinatorial algorithms\footnote{Informally, an 
algorithm that conducts a search over a discrete space. A key
distinction is that whereas a continuous method gives a way of solving the LP or convex program underlying a given instance, and provides 
no insight about the problem itself, a combinatorial algorithm handles all instances of the specific problem by exploiting its special 
combinatorial structure, e.g., see \cite{va.Nisan} for a detailed argument.}  
for exactly solving nonlinear convex programs. For this purpose, the classical primal-dual paradigm was suitably
extended from its usual setting of LP-duality theory to convex programs and the KKT conditions; this task was initiated in \cite{DPSV}. 

We note that all problems attacked so far were total, i.e., their
convex programs always have finite optimal solutions. The main contribution of this paper is to develop the methodology for handling problems that 
are not guaranteed to always have a solution, i.e., their convex programs may be infeasible for certain settings of the parameters. 
It turns out that in the setting of LP's, the primal-dual algorithms for non-total problems are not more involved than those for total
problems. We illustrate this by comparing the algorithms for the problems of maximum weight perfect matching and
maximum weight matching in bipartite graphs in Section \ref{sec.postmortem}. However, the situation is quite different for convex programs; 
see Section \ref{sec.feasibility} for a high level description and Section \ref{sec.postmortem} for a detailed analysis. 

The specific question that led to these ideas was combinatorially solving Nash bargaining games.
Nash bargaining \cite{nash.bargain} is a central solution concept within game theory for ``fair'' allocation of utility among competing players
in the presence of complete information; it has numerous applications and a large following, e.g., see \cite{Kalai.survey,Thomson.bargain,OR.GT}.
The solution to a Nash bargaining game is obtained by maximizing a concave function over a convex set, i.e., it is the solution to a convex 
program. If the conditions for efficiently running the ellipsoid algorithm hold \cite{GLS}, in particular,
if efficient separation oracles can be implemented for its constraints and objective function, then
the solution can be obtained to any required degree of accuracy in polynomial time. Instead, we resorted to the design of 
combinatorial algorithms because they have several advantages over ``continuous'' algorithms, see Section \ref{sec.rational}. 

The specific Nash bargaining game we study in this paper is called the {\em Arrow-Debreu Nash Bargaining Game}, abbreviated \RNB. 
The setup is the same as the linear case of the Arrow-Debreu market model, but instead of resorting to the solution concept of 
a market equilibrium for reallocating goods among the agents, we resort to a Nash bargaining solution. Interestingly enough, 
our algorithm for solving \RNB\ reduces it to a new, natural market model, which we call the {\em flexible budget market}.

In the next section, we first provide a formal, mathematical context to the new algorithmic activity mentioned above.
This also enables us to identify a number of questions for further study, see Section \ref{sec.discussion}.

\subsection{Rational convex programs and combinatorial algorithms for them}
\label{sec.rational}

A combinatorial algorithm for exactly solving a convex program is possible only if it admits rational solutions.
Starting with the classic Eisenberg-Gale convex program \cite{eisenberg},
whose solution yields an equilibrium for the linear case of Fisher's market model \cite{DPSV}, in recent years, researchers have identified
several interesting convex programs that always admit rational solutions. First, let us formally define this class.
We will say that a nonlinear convex program is {\em rational} if, for any setting of its parameters to rational numbers
such that it has a finite optimal solution, it admits an optimal solution that is rational and can be written using polynomially many bits
in the number of bits needed to write all the parameters.

This class turns out to be surprisingly rich -- besides several Nash bargaining games \cite{va.LNB2},
it captures the Fisher market model under several classes of utility functions: linear \cite{eisenberg, DPSV}, 
utility functions defined via combinatorial problems, including some in Kelly's \cite{kelly} resource allocation
model \cite{JV.EG,CDV.EG}, spending constraint utilities \cite{va.spending,Devanur-CP}, and piecewise-linear concave utilities in a market model
that allows for perfect price discrimination \cite{GV.discrimination}. It also captures the linear case of the Arrow-Debreu market model 
(this follows from the convex program of \cite{JainAD} and a proof of rationality in \cite{Eaves}). 

Of course, rationality does not guarantee the existence of a combinatorial polynomial time algorithm. However, it turns out that
such algorithms are known for solving almost all of the rational convex programs mentioned above. 
The only exceptions are 2-agent markets and Nash bargaining games given in \cite{CDV.EG} and \cite{va.LNB2}, and 
the linear case of Arrow-Debreu markets. For the former, polynomial time algorithms are given, in \cite{CDV.EG} and \cite{va.LNB2}, 
using only an LP-solver. This leads to a tantalizing question: is the class of rational convex programs solvable in this manner, i.e., by
polynomial time algorithms that use only an LP-solver?

Both algorithmic approaches, continuous and combinatorial, are valuable in their own right and have different advantages.
Indeed, it is the synergy between these two approaches that makes algorithm design a potent field. We point out the advantages
of the latter approach below.

We believe that the case for designing combinatorial algorithms for rational convex programs is as compelling 
as that for integral LP's, i.e., linear programs that always admit integral optimal solutions. 
The latter led to deep combinatorial insights and constitute a substantial part of the field of combinatorial optimization.
In recent years, insights gained from combinatorial algorithms for convex programs have also led to major progress. 
For instance, the recent proof of membership in PPAD of markets under piecewise-linear concave utilities \cite{VY.plc} followed from a 
new, combinatorial way of characterizing equilibria \cite{DPSV} and helped settle, together with \cite{Chen.plc,ChenTeng}, the long-standing 
open problem of determining the exact complexity of this key case of markets.

Combinatorial algorithms also have several advantages over continuous algorithms in applications. For instance, recently Nisan et. al., faced 
with the problem of designing an auction system for Google for TV ads, converged to a market equilibrium
based method, after exploring several different options \cite{Nisan.TV}. As stated by Nisan \cite{Nisan}, the actual implementation of this 
algorithm was inspired by combinatorial market equilibrium algorithms, which in turn solve convex programs combinatorially. Indeed, 
the easy adaptability of combinatorial algorithms to the special idiosyncrasies of an application often makes them the preferred method.
In Section \ref{sec.appl} we present an application of our combinatorial algorithm for \RNB \ to a throughput allocation problem on a wireless channel.

Although our algorithm is not strongly polynomial, recent results provide a reason to believe that our methodology could lead to it.
Using scaling techniques, Orlin \cite{Orlin} extended the DPSV algorithm to a strongly polynomial one, and an obvious open
question is to extend the rest of the algorithms for rational convex programs as well.

\subsection{Determining feasibility of a convex program}
\label{sec.feasibility}

A linear program is infeasible iff no point satisfies all its constraints. In the case of a convex program, infeasibility could arise
because of a second reason, namely the objective function is undefined at each point that satisfies all constraints. 
The convex program for \RNB \ exhibits only the second type of infeasibility; it always has points satisfying all constraints. 

Even for the case of total problems, the primal-dual method operates in fundamentally different ways in the two settings of integral
linear programs and rational convex programs; for details, see Section 4 in \cite{DPSV} or Section \ref{sec.postmortem} of the present paper. 
This difference manifests itself even more acutely when we move to non-total problems, as shown in detail in Section \ref{sec.postmortem}.

For solving \RNB, we give a special procedure, Stage I in Algorithm \ref{alg.main}, that tests for feasibility. 
If the instance is found to be infeasible, this procedure finds a proof of infeasibility. Otherwise, it yields special prices, 
called feasible prices, which provide Stage II with the right starting point for finding a solution to the given instance. 

We give two different proofs of infeasibility. The first one uses an explicit dual of the convex program for \RNB, found by 
Devanur \cite{Devanur.dual}; for proving infeasibility, our procedure shows that this dual is unbounded. 
Often it is difficult to find an explicit dual of a convex program. For this reason, we give a second proof
of infeasibility. We give an LP whose optimal solution is positive iff the given instance is feasible. Therefore, a dual feasible solution
of value at most zero is a proof of infeasibility. Our procedure also finds such a proof of infeasibility.

\subsection{New insights into balanced flows}
\label{sec.role-bf}

Differences between complementary slackness conditions for an LP and the KKT conditions for a convex program 
(see Section 4 in \cite{DPSV}) leads to new difficulties in the latter setting. The new algorithmic idea of balanced flows, 
introduced in \cite{DPSV}, helps overcome these difficulties. Although this notion yields the desired efficient algorithm,  
the use of the $\ltwo$ norm in the potential function used for proving polynomial time termination makes the proofs quite difficult. 
\cite{DPSV} observe that the $\ltwo$ norm can be dispensed with for 
defining balanced flows and ask, ``Can a polynomial running time be established for the algorithm using the 
alternative definition, thereby dispensing with the $l_2$ norm altogether? 
At present we see no way of doing this ... .''

In Section \ref{sec.tight} we answer this question in the negative by providing an infinite family of examples
in which the natural $\lone$ norm-based potential function makes inverse exponentially small progress whereas the $\ltwo$ 
norm-based potential function makes inverse polynomial progress.

We observe that the combinatorial algorithms of \cite{va.spending} and \cite{GV.discrimination}, which solve rational convex
programs, also use the notion of balanced flows crucially. Balanced flows play an even more fundamental role in our work. As detailed in
Section \ref{sec.use}, they are used in several crucial ways 
in both stages of our algorithm. In addition, they are also used for defining the central notion of feasible prices.

\section{The Nash Bargaining Game} 
\label{sec.Nash}

An {\em $n$-person Nash bargaining game} consists of a pair $(\CN, \cc)$, where $\CN \subseteq \Rplus^n$
is a compact, convex set and $\cc \in \CN$. Set $\CN$ is the {\em feasible set} and its elements give
utilities that the $n$ players can simultaneously accrue. Point $\cc$ is the {\em disagreement point}
-- it gives the utilities that the $n$ players obtain if they decide not to cooperate.
The set of $n$ agents will be denoted by $B$ and the agents will be numbered $1, 2, \ldots n$.
Game $(\CN, \cc)$ is said to be {\em feasible} if there is a point
$\vv \in \CN$ such that $\forall i \in B, \ v_i > c_i$, and {\em infeasible} otherwise.

The solution to a feasible game is the point $\vv \in \CN$ that satisfies
the following four axioms:

\begin{enumerate}
\item
{\bf Pareto optimality:}  No point in $\CN$ can weakly dominate $\vv$.
\item
{\bf Invariance under affine transformations of utilities:} If the utilities of any player are redefined by
multiplying by a scalar and adding a constant, then the solution to the transformed game is obtained by
applying these operations to the particular coordinate of $\vv$.
\item
{\bf Symmetry:} If the players are renumbered, then it suffices to renumber the coordinates of $\vv$ accordingly.
\item
{\bf Independence of irrelevant alternatives:} If $\vv$ is the solution for $(\CN, \cc)$, and 
$\CS \subseteq \Rplus^n$ is a compact, convex set satisfying $\cc \in \CS$ and  $\vv \in \CS \subseteq \CN$, 
then $\vv$ is also the solution for $(\CS, \cc)$.
\end{enumerate}

Via an elegant proof, Nash proved:

\begin{theorem}
{\bf Nash} \cite{nash.bargain}
\label{thm.nash}
If game $(\CN, \cc)$ is feasible
then there is a unique point in $\CN$ satisfying the axioms stated above.
This is also the unique point that maximizes $\Pi_{i \in B}  {(v_i - c_i)}$, over all $\vv \in \CN$.
\end{theorem}

Most papers in game theory assume that the given Nash bargaining game $(\CN, \cc)$ is feasible.
However, in this paper, it will be more natural to not make this assumption 
and to determine this fact algorithmically. Thus, we can have one of 2 outcomes:

Thus Nash's solution to his bargaining game involves maximizing a concave function over
a convex domain, and is therefore the optimal solution to the following convex program.

\begin{lp}
\label{CP-Nash}
\maximize & \sum_{i \in B}  \log (v_i - c_i)  \\[\lpskip]
\st       & \vv \in \CN     \nonumber  
\end{lp}

As a consequence, if for a specific game, a separation oracle can be implemented
in polynomial time, then using the ellipsoid algorithm one can get as good an approximation to
the solution of this convex program as desired in time polynomial in the number of bits
of accuracy needed \cite{GLS}.

\section{The Game ADNB}
\label{sec.LNB}

The game \RNB, short for {\em Arrow-Debreu Nash Bargaining game},
which will be studied extensively, is derived from the linear case of the Arrow-Debreu model.
We state the latter first.
Let $B = \{1, 2, \ldots, n\}$ be a set of agents and $G = \{1, 2, \ldots, g\}$ be a set of divisible goods.
We will assume w.l.o.g. that there is a unit amount of each good. 
Let $u_{ij}$ be the utility derived by agent $i$ on receiving one unit of good $j$; w.l.o.g., we
will assume that $u_{ij}$ is integral.
If $x_{ij}$ is the amount of good $j$ that agent $i$ gets, for $1 \leq j \leq g$, then she derives total utility 
\[ v_i(x) = \sum_{j \in G}  {u_{ij} x_{ij} } .\]
Finally, we assume that each agent has an initial endowment of these goods; for each good, the total amount 
possessed by the agents is 1 unit. 

W.l.o.g. we may assume that each good is desired by at least one agent and each agent
desires at least one good, i.e.,
\[ \forall j \in G, \ \exists i \in B: \ u_{ij} >0 \ \ \mbox{and}  \ \
\forall i \in B, \ \exists j \in G: \ u_{ij} >0 .\]
If not, we can remove the good or the agent from consideration.

The question is to find prices for these goods so that if each agent sells her entire initial
endowment at these prices and uses the money to buy an optimal bundle of goods, the market clears exactly, 
i.e., there is no deficiency or surplus of any good. Such prices are called {\em equilibrium prices}.

The Arrow-Debreu market model gives one mechanism by which the agents can redistribute goods to achieve 
higher utilities. Another mechanism is to view this setup as a Nash bargaining game as follows.
For each $i \in B$, let $c_i$ denote the utility derived by agent $i$ from her initial endowment.
Regard this as agent $i$'s disagreement utility and 
redistribute the goods in accordance with the Nash bargaining solution.

We will define game \RNB\ in a slightly more general manner: we will assume that the
disagreement utilities, $c_i$'s, are arbitrary numbers specified in the particular instance.
(As stated in the Introduction, we will not deal with the nonsymmetric extension of \RNB\ in this paper.)
Clearly, the Nash bargaining solution is the optimal solution to the following convex program:

\begin{lp}
\label{CP-ADNB}
\maximize & \sum_{i \in B}  \log (v_i - c_i)  \\[\lpskip]
\st       & \forall i\in B:~  v_i = \sum_{j \in G} {u_{ij} x_{ij} }  \nonumber  \\
          & \forall j\in G:~ \sum_{i\in B} x_{ij} \leq 1  \nonumber \\
          & \forall i \in B, \ \forall j\in G:~ x_{ij} \geq 0 \nonumber
\end{lp}

The KKT conditions for this program are:

\begin{description}
\item {\bf (1)}
$ \forall j \in G: ~ p_j \geq 0$.

\item {\bf (2)}
$ \forall j \in G: ~ p_j > 0  \  \implies \  \sum_{i\in B} {x_{ij}} = 1$.

\item [(3)]
$ \forall i \in B, \ \forall j \in G: p_j \geq {u_{ij} \over {v_i - c_i}}$.

\item [(4)]
$ \forall i \in B, \ \forall j \in G:  x_{ij} > 0 \ \implies p_j = {u_{ij} \over {v_i - c_i}}$.
\end{description}

\begin{theorem}
\label{thm.bits}
Program (\ref{CP-ADNB}) is a rational convex program. Moreover, if it is feasible, then the dual solution
is unique.
\end{theorem}

\begin{proof}
We will show that if for a setting of rational parameters, program (\ref{CP-ADNB}) is feasible, then the $x_{ij}$'s and $p_j$'s are 
solutions to an LP and are therefore rational numbers that can be written using polynomially many bits. 
First, guess the $x_{ij}$'s that are non-zero in the optimal solution to program (\ref{CP-ADNB}). Because of the assumption made on the
instance, each $p_j$ will be positive.

The variables of the LP will be the non-zero $x_{ij}$'s and for each good $j$, a new variable $q_j$, which is supposed to represent $1/p_j$. 
The LP will have the following constraints: for each $q_j$, there is one equation corresponding
to the KKT condition (2), and for each nonzero $x_{ij}$ there is one equation corresponding to
the KKT condition (4). In addition, the LP has inequality constraints corresponding to KKT condition (3), for each
$i \in B$ and $j \in G$. In all these constraints, $v_i$ is replaced by $\sum_{j \in G} {u_{ij} x_{ij}}$. Finally, it has non-negativity
constraints for all $x_{ij}$'s and $q_j$'s. It is easy to check that all constraints are linear.

Since program (\ref{CP-ADNB}) is feasible, so is the LP corresponding to the correct guess.
The solution to this LP will satisfy all KKT conditions and hence is an optimal solution to program (\ref{CP-ADNB}). This established
rationality of the convex program. Finally, the
strict concavity of the objective function of program (\ref{CP-ADNB}) and the fact that the convex combination of any two of its feasible
solutions is also feasible implies that the optimal values of $v_i$'s is unique. Now, using the KKT condition (4), we get the uniqueness of
$p_j$'s as well.
\end{proof}

\section{An Application to Throughput Allocation\\ 
on a Wireless Channel}
\label{sec.appl}

A central throughput allocation problem arising in the context of a wireless channel, such as in 3G technologies, is the following. 
There are $n$ users $1, 2, \ldots n$ and the 
wireless router can be in any of $m$ different states $1, 2, \ldots m$ whose probabilities, $\pi(j)$, can be estimated by sampling.
Each user $i$ derives utility at rate $u_{ij}$ if it is connected to the
router while the router is in state $j$; the $u_{ij}$'s are known. No matter what state the router is in, only one user can be
connected to it. If user $i$ is given connection for $x_{ij} \leq \pi(j)$ of the time the router is  
in state $j$, for $1 \leq j \leq m$, then the total utility derived by $i$ is
$v_i =  \sum_{j = 1}^m {u_{ij} x_{ij}}$.
Clearly, we must ensure the constraint $\sum_{i = 1}^n {x_{ij}}  \leq \pi(j)$, for each $j$.
The question is to find a ``fair'' way of dividing the $\pi(j)$'s among the users.

The method of choice in the networking community is to use Kelly's proportional fair scheme \cite{kelly}, which entails maximizing
$\sum_i {\log v_i}$ subject to the constraints given above, i.e., solving the Eisenberg-Gale convex program \cite{eisenberg}.
Observe that the above setting can be viewed as a linear Fisher market with $n$ users and $m$ divisible goods. 
An elegant gradient descent algorithm for solving this convex program, given by David Tse \cite{Tse} (see also \cite{JalaliRP00}), 
was implemented by Qualcomm in their chip sets and is used by numerous 3G wireless basestations \cite{Andrews}.
However, this solution may at times allocate unacceptably low utility to certain users. This was countered by giving users the ability to put 
a lower bound on channel rates, say $c_i$ for user $i$. This enhanced problem was solved by changing the objective function of the convex 
program to maximizing $\sum_i {\log (v_i - c_i)}$; observe that this is precisely an instance of \RNB! However, now the gradient descent
implementation ran into problems of instability, since it involved computing $u_{ij}/(v_i' - c_i)$, where $v_i'$ is the current estimate of $v_i$;
at intermediate points, the denominator may be too small or even negative. 

A different solution, proposed and implemented by researchers at Lucent \cite{AQS}, was to introduce the constraints $v_i > c_i$ in the 
Eisenberg-Gale program itself. The fairness guarantee achieved by this solution is unclear. Additionally, determining feasibility
of the convex program now became a major issue \cite{Andrews}. 
Instead, we have proposed experimenting with a heuristic adaptation of our 
combinatorial algorithm for \RNB, which will not have stability issues. As reported in the FOCS 2002 version of \cite{DPSV}, an analogous 
heuristic adaptation of the DPSV algorithm was found to perform well on fairly large sized linear Fisher instances.

\section{Fisher's Model and its Extension via Flexible Budgets}
\label{sec.model}

First we specify Fisher's market model for the case of linear utilities \cite{scarf}. 
Consider a market consisting of a set of $n$ buyers 
$B = \{1, 2, \ldots, n\}$, and a set of $g$ divisible goods, $G = \{1, 2, \ldots, g\}$; 
we may assume w.l.o.g. that there is a unit amount of each good. Let $m_i$ be the money 
possessed by buyer $i$, $i \in B$.
Let $u_{ij}$ be the utility derived by buyer $i$ on receiving one unit of good $j$. 
Thus, if $x_{ij}$ is the amount of good $j$ that buyer $i$ gets, for $1 \leq j \leq g$, then the total 
utility derived by $i$ is
\[ v_i(x) = \sum_{j=1}^g  {u_{ij} x_{ij} } .\]

The problem is to find prices $\pp = \{p_1, p_2, \ldots , p_g \}$ for the goods so that when each
buyer is given her utility maximizing bundle of goods, the market clears, i.e., each
good having a positive price is exactly sold, without there being any deficiency or surplus. 
Such prices are called {\em market clearing prices} or {\em equilibrium prices}.

The following is the Eisenberg-Gale convex program. Using the KKT conditions, one can show that its
optimal solution is an equilibrium allocation for Fisher's linear market and the Lagrange variables
corresponding to the inequalities give equilibrium prices for the goods 
(e.g., see Theorem 5.1 in \cite{va.chapter}).

\begin{lp}
\label{CP-EG}
\maximize & \sum_{i \in B} {m_i \log v_i }  \\[\lpskip]
\st       & \forall i\in B:~  v_i = \sum_{j \in G} {u_{ij} x_{ij} }  \nonumber  \\
          & \forall j\in G:~ \sum_{i\in B} x_{ij} \leq 1  \nonumber \\
          & \forall i \in B, \ \forall j\in G:~ x_{ij} \geq 0 \nonumber
\end{lp}

Next, we introduce a flexible budget market as a modification of Fisher's linear case; this market will be used for solving \RNB.
The utility functions of buyers are as before.
The two main differences are that each buyer $i$ now has a parameter $c_i$ giving
a strict lower bound on the amount of utility she wants to derive and buyers do not come to the market
with a fixed amount of money, but instead the money they spend is a function of prices of goods in the
following manner.
Given prices $\pp$ for the goods, define the {\em maximum bang-per-buck} of buyer $i$ to be
\[ \ga_i = max_j \left\{ {u_{ij} \over p_j} \right\} .  \]
Now, buyer $i$'s money is defined to be $m_i = 1 + {c_i \over \ga_i}$.

Let us say that set $S_i = \argmax_j  \left\{ {u_{ij} \over p_j} \right\}$
constitutes $i$'s {\em maximum bang-per-buck goods}. Clearly, at prices $\pp$, any utility maximizing
bundle of goods for $i$ will consist of goods from $S_i$ costing $m_i$ money. Again the problem is to find 
market clearing or equilibrium prices. Observe that in an equilibrium, if it exists, each buyer $i$ will 
derive utility exceeding $c_i$.

\subsection{The Reduction}
\label{sec.reduce}

An instance $I$ of \RNB \ is transformed to a flexible budget market $\CM$ as stated above.

\begin{theorem}
\label{thm.reduce}
Instance $I$ is feasible iff $\CM$ is feasible.
Moreover, if $I$ and $\CM$ are both feasible, then allocations $\xx$ and dual $\pp$ are optimal 
for $I$ iff they are equilibrium allocations and prices for the flexible budget market $\CM$.
\end{theorem}

\begin{proof}
($\Rightarrow$)
First assume that $I$ is feasible and that allocations $\xx$ and dual $\pp$ are optimal for RNB game $I$.
Then $I$ must satisfy the KKT conditions for convex program (\ref{CP-ADNB}).

By the second KKT condition, each good having a positive price is fully sold. 
Assume that $y_{ij} > 0$. Then, by the definition of $\ga_i$ and the fourth KKT condition,
\[ \ga_i = {u_{ij} \over p_j} = {v_i - c_i} .\]
The money of buyer $i$ at prices $\pp$ in market $\CM$ is defined to be $m_i = 1 + c_i/\ga_i$.
The money spent by $i$ in market $\CM$ is:
\[ \sum_{j \in G} {x_{ij} p_j} =  \sum_{j \in G} {{x_{ij} u_{ij}} \over {\ga_i}} \]
\[ = {1 \over {v_i - c_i}} \sum_{j \in G} {x_{ij} u_{ij}} = {{c_i} \over {v_i - c_i}} 
= 1 + {{c_i} \over {v_i - c_i}} = 1 + {c_i \over {\ga_i}} = m_i .\]
Furthermore, by the third and fourth KKT conditions, $i$ buys only her maximum bang-per-buck objects,
thereby getting an optimal bundle. This proves that $\xx$ and $\pp$ constitute equilibrium allocations 
and prices for market $\CM$.

($\Leftarrow$)
Next, assume that $\CM$ is feasible and that $\xx$ and $\pp$ are equilibrium allocations and prices for 
market $\CM$. Now, $\xx$ is clearly feasible for program (\ref{CP-ADNB});
we will show that $\xx$ and $\pp$ satisfy all the KKT conditions for this program.
The first two conditions are obvious.

Since $i$ gets an optimal bundle of objects at prices $\pp$,
\[ x_{ij} > 0 \  \Rightarrow \ {u_{ik} \over p_j} = \ga_i .\]
Since $i$ spends all her money,
\[ m_i = 1 + {c_i \over \ga_i} = \sum_{j \in G} {x_{ij} p_j} 
= \sum_{k \in T_i} {y_{ik} {u_{ik} \over \ga_i}} = {v_i \over \ga_i} .\]
Therefore, $\ga_i = {v_i - c_i}$.
This gives the last two conditions as well.
\end{proof}

\section{A Test for Equilibrium Prices}
\label{sec.test}

We will present an efficient algorithm for solving an instance $I$ of the game
\RNB \ by first reducing it to a flexible budget market, say $\CM$.
In this section, we first give an efficient algorithm for the following simpler question:
Given prices $\pp = \{p_1, \ldots p_g \}$ for the goods in $\CM$, determine if these are equilibrium prices, and
if so, find an equilibrium allocation.

First, construct a directed network $N(\pp)$ as follows.
$N(\pp)$ has a source $s$, a sink $t$, and vertex subsets $B$ and $G$ corresponding to the buyers and goods,
respectively. For each good $j \in G$, there is an edge $(s, j)$ of capacity $p_j$, 
and for each buyer $i \in B$, there
is an edge $(i, t)$ of capacity $m_i$, where $m_i = 1 + c_i/\ga_i$ is $i$'s money in $\CM$.
Recall that $S_i$ contains $i$'s maximum bang-per-buck goods.  The edges between $G$ and $B$ are precisely the 
maximum bang-per-buck edges, i.e., those $(j, i)$ such that $j \in S_i$. Each of these edges has infinite capacity. 

\begin{lemma}
\label{lem.N}
Prices $\pp$ are equilibrium prices for $\CM$ iff the two cuts $(s, B \cup G \cup t)$ and
$(s \cup B \cup G, t)$ are min-cuts in network $N(\pp)$. 
Moreover, if $\pp$ are equilibrium prices, then the set of equilibrium allocations corresponds exactly to 
max-flows in $N(\pp)$.
\end{lemma}

The proof of this lemma is straightforward using the transformation between a max-flow $f$ in $N(\pp)$
and an allocation $x$ in $\CM$ given by $x_{ij} = f(j, i)/p_j$. The condition that 
$(s, B \cup G \cup t)$ and $(s \cup B \cup G, t)$ are min-cuts in network $N(\pp)$, and hence 
saturated by $f$, corresponds to all goods being sold and all buyers' money being spent. 
The fact that $(j, i)$ is an edge in $N(\pp)$ iff $j \in S_i$ ensures that buyers get only their 
maximum bang-per-buck goods. Clearly, one max-flow computation suffices to determine if prices $\pp$ are equilibrium prices for $\CM$.

The next lemma gives the combinatorial object that yields equilibrium prices.
Assume that $\pp^*$ are equilibrium prices, i.e., $N(\pp^*)$ satisfies the condition in Lemma \ref{lem.N}. Let $H$ be the 
uncapacitated directed subgraph of $N(\pp^*)$ induced on $B \cup G$.

\begin{lemma}
\label{lem.discrete}
Given $H$, we can find $\pp^*$ in strongly polynomial time.
\end{lemma}

\begin{proof}
Consider the connected components of $H$ after ignoring directions on its edges. In each component, pick a good and assign it
price $p$, say. The prices of the rest of the goods in this component can be obtained in terms of $p$. The bang-per-buck, and hence
the money, of each buyer in this component can also be obtained in terms of $p$. Finally, by equating the money of all buyers in this component
with the total value of all goods in this component, we can compute $p$.
\end{proof}

Given two $d$-dimensional vectors with non-negative coordinates, $\pp$ and $\qq$, we will say that $\pp$ {\em weakly dominates} 
$\qq$ if for each coordinate $j$, $q_j \leq p_j$. 
The following characterization will be useful in Algorithm \ref{alg.limit}.

\begin{lemma}
\label{lem.bounded}
Assume that market $\CM$ is feasible and $\pp$ is its unique equilibrium price vector. Let $\qq$ be a vector of positive prices such that
$(s, B \cup G \cup t)$ is a min-cut in network $N(\qq)$. Then $\pp$ weakly dominates $\qq$.
\end{lemma}

\begin{proof}
Let us assume, for establishing a contradiction, that there are goods $j$ such that $q_j > p_j$ and yet $(s, B \cup G \cup t)$ is a 
min-cut in $N(\qq)$. Let 
\[  \theta = \max_{j \in G} \left\{ {q_j \over p_j} \right\} \ \  \ \ \mbox{and} \ \ \ \ 
S = \{j \in G ~ | ~ q_j = \theta p_j \} . \]
Clearly, $\theta >1$.

Let $T_p$ and $T_q$ be the set of buyers who are interested in goods in $S$ at prices $p$ and
$q$, respectively. Since $S$ represents the set of goods whose prices increase by the largest 
factor in going from prices $\pp$ to $\qq$, $T_q \subseteq T_p$. 
We claim that a buyer $i \in T_q$ is not interested in any goods in $G - S$ at prices $\pp$ (because
otherwise at prices $\qq$, $i$ will not be interested in any goods in $S$, since their prices increased
the most). Therefore, in any max-flow in $N(\pp)$, all flow going through nodes in $T_q$ must also have
used nodes in $S$. Therefore,
\[ \sum_{j \in S}  {p_j}  \geq  \sum_{i \in T_q} \left(1 + {c_i \over \ga_i} \right) ,\]
where $\ga_i$ is the maximum bang-per-buck of buyer $i$ w.r.t. prices $\pp$.
Multiplying this inequality by $\theta$ and using the fact that $\theta > 1$ we get,
\[ \theta \sum_{j \in S}  {p_j}  \geq  \sum_{i \in T_q} \left(\theta + {{c_i \theta} \over \ga_i} \right) 
> \sum_{i \in T_q} \left(1 + {{c_i \theta} \over \ga_i} \right) ,\]

Observe that the maximum bang-per-buck of buyer $i \in T_q$ w.r.t. prices $\qq$ is $\ga_i / \theta$.
Therefore, the last inequality implies that w.r.t. prices $\qq$, the total value of goods in $S$ is strictly 
more than the total value of money possessed by buyers in $T_q$.
On the other hand, since in $N(\qq)$ all flow using nodes of $T_q$ goes through $S$, we get that 
$(s, B \cup G \cup t)$ is not a min-cut in network $N(\qq)$, leading to a contradiction.
\end{proof}

Next assume that $\CM$ is an arbitrary flexible budget market, not necessarily feasible. We will say that prices $\qq$ are {\em small}
if $\qq$ is a positive vector and $(s, B \cup G \cup t)$ is a min-cut in network $N(\qq)$. Observe that in this case, 
each good $j$ must have an edge $(j, i)$, for some buyer $i$, incident at it. By Lemma \ref{lem.bounded}, if $\CM$ is feasible and 
prices $\qq$ are small, then they are weakly dominated by the equilibrium prices, $\pp$. Observe however that the contrapositive of 
Lemma \ref{lem.bounded} does not hold, i.e., $\pp$ may dominate positive prices $\qq$, yet 
$(s, B \cup G \cup t)$ may not be a min-cut in network $N(\qq)$.

The background given so far will suffice to read Section \ref{sec.limit}, which gives an algorithm that 
converges to the solution of a given feasible instance of \RNB \ in the limit. 
This section may serve as a suitable warm-up, since our polynomial time algorithm is quite involved.

\section{Determining Feasibility}
\label{sec.high}

In this section, we will address the question of determining whether the given flexible budget market, $\CM$, is feasible.
We will give a characterization of feasible markets and we will derive conditions that yield a proof of infeasibility.

Clearly, if we can find small prices $\pp$ and a max-flow $f$ in network $N(\pp)$ such that the flow
gives each buyer $i$ strictly more than $c_i$ utility, then $\CM$ is feasible. We first show, using the notion of balanced
flows, that this test of feasibility is in fact a property of prices $\pp$ only.

\subsection{Balanced flows}
\label{sec.balanced-flow}

We will follow the exposition in \cite{va.chapter} and refer
the reader to this chapter for all facts stated below without proof. 
For simplicity, denote the current network, $N(\pp)$, by simply $N$.
Given a feasible flow $f$ in $N$, let $R(f)$ denote
the residual graph w.r.t. $f$. Define the {\em surplus} of buyer $i$ w.r.t. flow $f$ in network $N$, $\th_i(N,f)$, 
to be the residual capacity of the edge $(i,t)$ w.r.t. flow $f$ in network $N$, i.e.,
$m_i$ minus the flow sent through the edge $(i,t)$. 
The {\em surplus vector w.r.t. flow f} is defined to be $\th(N,f) := (\th_1(N,f),\th_2(N,f),\ldots, \th_n(N,f))$.
Let $\|v\|$ denote the $l_2$ norm of vector $v$. A {\em balanced flow} in network $N$ is a flow
that minimizes  $\|{\th(N,f)}\|$. A balanced flow must be a max-flow
in $N$ because augmenting a given flow can only lead to a decrease in the $l_2$ norm of the surplus vector.

A balanced flow in $N$ can be computed using at most $n$ max-flow computations.
It is easy to see that all balanced flows in $N$ have the same surplus vector. Hence, for each buyer $i$, we can define
$\th_i(N)$ to be the surplus of $i$ w.r.t. any balanced flow in $N$; we will shorten this to $\th_i$ when the network is understood. 
The key property of a balanced flow that our algorithm will rely on is that a maximum flow $f$ in $N$ is balanced iff it satisfies Property 1:

\noindent
{\bf Property 1:} For any two buyers $i$ and $j$, if $\th_i(N, f) < \th_j(N, f)$ then there is no 
path from node $i$ to node $j$ in $R(f)-\{s, t\}$.

Balanced flows play a crucial role in both stages of our algorithm; moreover, they have multiple uses.
In Section \ref{sec.use}, after stating the full algorithm, we state the various uses of this notion.

\subsection{A characterization of feasibility}
\label{lem.cond-feasible}

Let $\pp$ be small prices and let $(\theta_1, \ldots, \theta_n)$ be the surplus vector of a balanced flow in $N(\pp)$.
We will say that $\pp$ are {\em feasible prices} if for each buyer $i$, $\theta_i <1$.

\begin{lemma}
\label{lem.feasible-p}
Market $\CM$ is feasible iff there are feasible prices for it.
\end{lemma}

\begin{proof}
If $\CM$ is feasible, its equilibrium prices are feasible, since for each buyer $i$, $\th_i = 0$.
Next, assume that $\pp$ are feasible prices for $\CM$. By definition, the flow sent on edge $(i, t)$ in 
a balanced flow in $N(\pp)$ is 
\[  m_i - \th_i > m_i -1 = \left(1 + {c_i \over \ga_i}\right) - 1 =  {c_i \over \ga_i} .\]
The utility accrued by $i$ from this allocation is $\ga_i (m_i - \th_i) > c_i$.
Hence $\CM$ is feasible.
\end{proof}

Observe that if a max-flow $f$ in network $N(\pp)$ 
gives each buyer $i$ strictly more than $c_i$ utility, then so will a balanced flow in $N(\pp)$. Hence, feasibility is a 
property of prices $\pp$ only.

Rather than working with $\th_i$, it will sometimes be more convenient to work with $\th_i -1$. Hence, w.r.t. small prices $\pp$,
let us define the {\em 1-surplus} of buyer $i$ to be $\bt_i = \th_i - 1$, where 
$(\theta_1, \ldots, \theta_n)$ is the surplus vector of a balanced flow in $N(\pp)$. 
Now, another definition of feasible prices is that they be small and for each buyer $i$, $\bt_i < 0$.

\subsection{Proof of infeasibility}
\label{lem.cond-infeasible}

We will provide two ways of establishing infeasibility of the given market. 
The first is via the dual of an LP whose optimal solution tells us if $\CM$ is feasible
and the second is via the dual of convex program (\ref{CP-ADNB}).
Given an infeasible market, Stage I of our algorithm will terminate in a way that yields both proofs.

The game is feasible iff there is a point $v \in \CN$ such that for each agent $i \in B$, $v_i > c_i$.
In order to capture feasibility via a linear program, let us restate as follows: the game is feasible iff
\[  \max_{v \in \CN} \  \min_{i \in B} : \ (v_i - c_i)  > 0 .\]
Observe that the expression on the left hand side is the optimal objective function value of LP (\ref{LP-maxt}).
Hence, the game is feasible iff the optimal solution to LP (\ref{LP-maxt}) is greater than zero.
Clearly, this LP is maximizing $t$; however, in order to obtain a convenient dual, we will write it as minimizing $-t$:

\begin{lp}
\label{LP-maxt}
\minimize &    -t    \\[\lpskip]
\st       & \forall i\in B:~   \sum_{j \in G} {u_{ij} x_{ij}}  \geq  c_i + t  \nonumber  \\
          & \forall j\in G:~ - \sum_{i\in B} x_{ij} \geq -1  \nonumber \\
          & \forall i \in B, \ \forall j\in G:~ x_{ij} \geq 0 \nonumber
\end{lp}

Let $y_i$'s and $z_j$'s be the dual variables corresponding to the first and second set of inequalities,
respectively. The dual program is:

\begin{lp}
\label{LP-dualt}
\maximize & \sum_{i \in I} {c_i y_i}  -  \sum_{j \in G} {z_j}       \\[\lpskip]
\st       & \forall i \in B, \ \forall j\in G:~ u_{ij} y_i - z_j \leq 0   \nonumber \\
          &  \sum_{i \in B} {y_{i}}  =  1    \nonumber  \\
          & \forall i \in B:~ y_{i} \geq 0  \nonumber \\
          & \forall j \in G:~ z_{j} \geq 0  \nonumber
\end{lp}

\begin{lemma}
\label{lem.infeasible}
If there exist prices $\pp$ s.t. ${\sum_{j \in G} {p_j}} > 0$ and 
$\sum_{i \in B} {\bt_i} \ \geq \ 0$, then the given game is infeasible.
\end{lemma}

\begin{proof}
W.r.t. prices $\pp$, compute the maximum bang-per-buck, $\ga_i$, of each buyer $i$.
By definition of maximum bang-per-buck,
\[ \forall i \in B, \ \forall j \in G: \ \ \ \ \ga_i \geq {u_{ij} \over p_j} . \]
Let $\mu = \sum_{i \in B} {1/{\ga_i}}$. Since $\sum_{j \in G} {p_j} > 0, \ \mu > 0$.

Next, consider prices $\qq$, where for each $j \in G$, $q_j = p_j/\mu$.
Since all the prices have been scaled by the same factor, the network remains unchanged.
Clearly, the maximum bang-per-buck of buyer $i$ w.r.t. $\qq$ is $\ga_i' = \mu \ga_i$ and
\[ \forall i \in B, \ \forall j \in G: \ \ \ \ \ga_i' \geq {{u_{ij}} \over q_j} . \]

Let $y_i = 1/\ga_i'$, for $i \in B$, and $z_j = q_j$, for $j \in G$.
We will show that $(y, z)$ is a feasible solution for the dual LP (\ref{LP-dualt}).
The first set of inequalities is established by noting that
\[ \forall i \in B, \ \forall j \in G: \ \ \ \ \ga_i' \geq {u_{ij} \over q_j} \ \ \
\mbox{hence} \ \ \ u_{ij} y_i \leq z_j .\]
Next we show that the equality constraint holds:
\[ \sum_{i \in B} {y_i} =  \sum_{i \in B} {1 \over \ga_i'} = \left({1 \over \mu}\right) \cdot \sum_{i \in B} {1 \over \ga_i} = 1 .\]

Let $\alpha_i' = c_i/\ga_i'$ and let $\bt_i'$ be the 1-surplus of buyer $i$ w.r.t. prices $\qq$.
The objective function value of the dual solution $(y, z)$ is 
\[ \sum_{i \in B}  {c_i y_i}  -  \sum_{j \in G}  {z_j}  = 
\sum_{i \in B}  {\alpha_i'}  -  \sum_{j \in G}  {q_j}  = \ \sum_{i \in B} {\bt_i'} \ 
= \ \left({1 \over \mu}\right) \cdot \sum_{i \in B} {\bt_i} \geq 0 .\] 
Therefore, at optimality, $-t \geq 0$, i.e., $t \leq 0$, hence establishing infeasibility of the game. 
\end{proof}

Next, we derive a condition under which the following convex program, which is the dual of (\ref{CP-ADNB}), 
has an unbounded solution, hence proving infeasibility of the primal. This dual program was given by Devanur \cite{Devanur.dual}.
Its variables are $q_j$'s and $y_i$'s.

\begin{lp}
\label{CP-dual}
\minimize & \sum_{j \in G} {q_j} - \sum_{i \in B} {c_i y_i}  -  \sum_{i \in B} {\log(y_i)}       \\[\lpskip]
\st       & \forall i \in B, \ \forall j\in G:~ q_j \geq u_{ij} y_i    \nonumber 
\end{lp}

\begin{lemma}
\label{lem.infeasible-CP}
If there exist positive prices $\pp$ such that network $N(\pp)$ can be partitioned into two, induced on
$(B', G')$ and $((B - B'), (G - G'))$, where $B' \subset B$ and $G' \subset G$, such that 
$\sum_{i \in (B - B')} {\bt_i} \ \geq \ 0$ and $\forall i \in (B - B'), \ \forall j \in G': \ u_{ij} = 0$,
then the given game is infeasible.
\end{lemma}

\begin{proof}
Observe that 
\[ \sum_{i \in (B - B')} {\bt_i} \ \ = \ \left(\sum_{i \in (B - B')} {\th_i} \right) - |B - B'| \ 
= \left( \sum_{i \in (B - B')} {m_i} \right) - \left( \sum_{j \in (G - G')} {p_j} \right) - |B - B'| \]
\[ = \left( \sum_{i \in (B - B')} {c_i \over \ga_i} \right) -  \left( \sum_{j \in (G - G')} {p_j} \right) \ \geq 0 .\]

It is easy to see that setting $q_j$ to $p_j$ and $y_i$ to $1/\ga_i$ gives a feasible solution to program (\ref{CP-dual}). 
Multiply the prices of all goods in $(G - G')$ by $x$ and let $x \rightarrow \infty$. Observe that 
because of the condition $\forall i \in (B - B'), \ \forall j \in G': \ u_{ij} = 0$, no new edges will be introduced
in the network. Hence, the updated setting of $q_j$'s and $y_i$'s still yields a feasible solution to the dual. 

As $x \rightarrow 0$, 
for each $i \in (B - B')$, $\log(\ga_i) \rightarrow - \infty$. Also, $\sum_{j \in (G - G')} {q_j} - \sum_{i \in (B - B')} {c_i y_i}$
is either 0 or tends to $-\infty$. Hence the entire objective function tends to $-\infty$.
Therefore, the dual is unbounded and hence the primal is infeasible.
\end{proof}

\section{Details of the Algorithm for ADNB}
\label{sec.alg-main}

We will impose the following condition throughout; by Lemma \ref{lem.bounded}, it will 
ensure that prices are always small\footnote{The following power point presentations may make the algorithm easier to understand:\\
http://www.cc.gatech.edu/~vazirani/Waterloo1.ppt \\
http://www.cc.gatech.edu/~vazirani/Waterloo2.ppt}.

\noindent
{\bf Invariant:} W.r.t. current prices, $\pp$, $(s, B \cup G \cup t)$ is a min-cut in network $N(\pp)$. 

It is easy to see that the prices found by Initialization satisfy the Invariant.

Let $f$ be a balanced flow in $N(\pp)$. Since the Invariant is always maintained, for each buyer $i$, $\th_i \geq 0$ and
hence $\bt_i \geq -1$. In the algorithm, we will change prices of a well-chosen set $J$ of goods as follows.
Multiply the price of each good in $J$ by a variable $x$ and initialize $x$ to 1.
In Stage I, we will decrease $x$ and in Stage II we will raise $x$ until the next event happens.

In the next lemma, we will assume that $J = G$ and we will study how the 1-surplus of buyers changes as a function of $x$.
Define $x \cdot f$ to be the flow obtained by multiplying by $x$ the flow on each edge w.r.t. $f$. 
Let $\bt_i(x)$ denote $i$'s 1-surplus w.r.t. flow $x \cdot \pp$. Let $B'$ be the set of buyers having negative 1-surplus w.r.t. prices $\pp$.
If $B' = \emptyset$, define $b = \infty$ else define $b = \min_{i \in B'} \{ {{- 1} / {\bt_i}} \}$.
Observe that in both cases, $b > 1$.

\begin{lemma}
\label{lem.xf-max}
Flow $x \cdot f$ is a balanced flow in $N(x \pp)$ for $0 < x \leq b$, and
for each $i \in B, \ \bt_i(x) = x \bt_i$.
\end{lemma}

\begin{proof}
Since the Invariant holds and $f$ is a max-flow in $N(\pp)$, the cut $(s, J \cup I \cup t)$ is 
saturated by $f$, and hence by $x \cdot f$ in $N(x \pp)$. 
Next we show that $x \cdot f$ is a feasible flow in $N(x \pp)$, i.e., for each buyer $i \in B$, edge $(i, t)$ is not over saturated. 
Now, $\bt_i = \al_i - f(i, t)$. 
Therefore, the surplus on edge $(i, t)$ w.r.t. flow $x \cdot f$ is $1 + x (\al_i - f(i, t)) = 1 + x \bt_i \geq 0$ 
for $0 < x \leq b$. Hence, edge $(i, t)$ is not over saturated. Furthermore, $\bt_i(x) = x \bt_i$.
Finally, since $f$ satisfies Property 1 in $N(\pp)$, $x \cdot f$ satisfies it in $N(x \pp)$, thereby showing that it is a balanced flow.
\end{proof}

\subsection{Details of Stage I}
\label{sec.alg-main-I}

Algorithm \ref{alg.main} gives the pseudo code for Stage I. In this section, we give the subroutines used by this stage
and Section \ref{sec.pred} gives formal definitions of the predicates used in the While loops.
A run of Stage I is partitioned into {\em phases}, which are further partitioned into {\em iterations}.
In Stage I, an iteration ends when a new edge is added to the network. A phase ends either when the
condition of Step \ref{step.zero} holds or if for some $i \in I, \beta_i \geq 0$. 
In each iteration, the algorithm computes a balanced flow in the current network, $N(\pp)$. 

We establish the following notation.
For $J \subseteq G$, define $p(J) = \sum_{j \in J} {p_j}$ and
$\Ga(J) = \{i \in B ~|~ \exists j \in J \ s.t. \ (j, i) \in N(\pp)\}$.
Similarly, for $I \subseteq B$, define $\al(I) = \sum_{i \in I} {\al_i}$, $m(I) = \sum_{i \in I} {m_i}$ and
$\Ga(I) = \{j \in G ~|~ \exists i \in B \ s.t. \ (j, i) \in N(\pp)\}$.

The sets $B_c$ and $G_c$ denote the {\em current sets of buyers and goods} being considered by the algorithm.
These sets are initialized to $B$ and $G$, respectively. 
At any point the algorithm, $B = B_c \cup B'$ and $G = G_c \cup G'$, where $B'$ and $G'$ are the sets of
{\em adaptable buyers and goods}, respectively; their purpose is explained below.
$B'$ and $G'$ are both initialized to $\emptyset$. As the algorithm proceeds, buyers are moved from $B_c$ to $B'$ and
goods are moved from $G_c$ to $G'$.

The subroutines used in Stage I are:
\begin{itemize}
\item
{\bf Find sets(I):} 
Sets $I \subseteq B_c$ and $J \subseteq G_c$ are initialized as follows.
\[ I \la \  \arg\min_{i \in B_c} \{\bt_i\}  \ \ \ \  \mbox{and} \ \ \ \  J \la \ (\Ga(I) - \Ga(B_c - I)) . \]
Observe that $J$ consists of goods that are the maximum bang-per-buck goods of buyers in $I$ only.
All edges from goods in $G_c - J$ to buyers in $I$ are removed; this is justified in 
Lemma \ref{lem.remove} below.

\item
{\bf Update sets(I):}
Find the set, $I'$, of all buyers in $B_c - I$ such that there is a residual path from a buyer in  
$I$ to a buyer in $I'$. Update
\[  I \la (I \cup I') \ \ \ \  \mbox{and} \ \ \ \  J \la \ (\Ga(I) - \Ga(B_c - I)) . \]
All edges from goods in $G_c - J$ to buyers in $I$ are removed. Once again, this
is justified in Lemma \ref{lem.remove}.
\end{itemize}

Assume that $(j, i), \ j \in J, \ i \in (B_c - I)$ is the new edge added to $S_i$ in the current iteration. 
Observe that if $I' = \emptyset$, then all the flow from $j$, which was going to buyers in $I$ before the addition of this
edge, must go to $i$, since there is no residual path from $I$ to $i$. Accordingly, {\bf Update sets(I)} will move good $j$ from $J$ to
$G_c - J$. As soon as the prices of goods in $J$ are reduced by an infinitesimally small amount (by decreasing $x$)
buyers in $I$ will not be interested in good $j$ anymore. 

\begin{lemma}
\label{lem.neigh1}
In Stage I, at the start of each iteration, for each buyer $i \in I$ there is a good $j \in J$ such that
edge $(j, i)$ is in the network.
\end{lemma}

\begin{proof}
Since $\bt_i < 0$, the balanced flow must be sending flow on some edge $(j, i)$.
If an edge $(j, i')$ to a buyer $i' \in (B_c - I)$ is also present in the network, then
there will be a residual path from $i$ to $i'$, violating Property 1.
Therefore, there is no such edge and $j \in (\Ga(I) - \Ga(B_c - I))$, proving the lemma.
\end{proof}

We now explain the purpose of the sets $B'$ and $G'$. Once a good is moved into $G'$, its price gets frozen until the end of Stage I.
At any point in Stage I, these sets satisfy the following properties:
\begin{enumerate}
\item
W.r.t.\ the frozen prices of goods in $G'$, for each buyer $i \in B', \ \bt_i <0$.
\item
Buyers in $B$ are totally uninterested in goods in $G'$ at any price, i.e., 
for every $i \in B$ and every $j \in G'$, $u_{ij} = 0$. Hence, as prices of goods in $G'$
are decreased, no edge from $B$ to $G'$ will ever enter the network.
\end{enumerate}

The reason for the name ``adaptable'' is that as far as determining feasibility or infeasibility goes,
buyers in $B'$ can be made consistent with the outcome of the remaining buyers. Thus, if
$\forall i \in B_c, \ \bt_i < 0$, by assigning the frozen goods in $G'$ their prices at the time of freezing, we
can ensure that $\forall i \in B, \ \bt_i < 0$. The resulting price vector is clearly feasible. In Step \ref{step.restore}
in Algorithm \ref{alg.main} we have refer to this process as ``restoring prices of adjustable goods.'' 

If on the other hand $\sum_{i \in B_c} {\bt_i} \geq 0$, then we can give a proof of infeasibility in one of two ways.
First, by lowering the prices of all goods in $G'$ to zero, which can be done without introducing any new edges in the network, we
can ensure that $\forall i \in B', \ \bt_i = 0$, thereby ensuring that these buyers don't affect the sum of $\bt_i$'s. The conditions of 
Lemma \ref{lem.infeasible} now hold and yield a proof of infeasibility. Second, by assigning the frozen goods in $G'$ their prices 
at the time of freezing, we can ensure that the prices of all goods are positive and the conditions of 
Lemma \ref{lem.infeasible-CP} now hold to yield a different proof of infeasibility.

At the start of a phase, the set $I \subseteq B_c$ of buyers having the smallest $\bt$ values is identified. The goods they
desire are put in set $J$. If at any point, $I$ and $J$ are found to be adaptable, the algorithm updates $B'$ and 
$G'$ and the phase comes to an end. Otherwise, the algorithm lowers the prices of goods in $J$ until a new edge $(j, i)$,
with $j \in J$ and $i \in (B_c - I)$ is added to the network. On recomputing a balanced flow, either $\bt_i$ becomes negative, if so
$i$ moves into $I$ and the iteration comes to an end, or for some buyer(s) $i' \in I, \ \bt_{i'}$ increases. If
for some buyer $i \in I, \ \bt_{i}$ becomes non-negative, the phase comes to an end. Lemma
\ref{lem.StageI-terminate} shows that eventually, either all buyers are rendered good or the conditions of Lemma \ref{lem.infeasible} 
start holding. In the former case, the frozen prices of goods in $G'$ are restored and the algorithm moves on to Stage II
to find equilibrium prices.

One way to view the operation of Stage I is as a tug-of-war between two sets of buyers: the good buyers and the rest.
The algorithm decreases the prices of goods desired by buyers in $I$, thereby increasing their $\bt_i$'s. This helps towards
reaching the infeasibility condition stated above. However, as new edges enter the network and a balanced
flow is recomputed, buyers may move between the 2 sets. 

Observe that in each iteration, the algorithm needs to compute the largest value of $x$ at which a new edge is added to the network.
For any one edge this is straightforward; taking the maximum over all relevant edges gives the required value.

\begin{lemma}
\label{lem.StageI-terminate}
Stage I must terminate with either a feasible price vector or a proof of infeasibility.
\end{lemma}

\begin{proof}
By Lemma \ref{lem.StageI}, Stage I must terminate. If at this point, $\forall i \in B: \ \bt_i < 0$, a feasible price
vector has been found. Otherwise, $\sum_{i \in B_c} {\bt_i} \geq 0$ and hence $\exists i \in B_c: \ \bt_i \geq 0$. 
This buyer must be in $B_c - I$, since all buyers in $I$ satisfy $\bt_i < 0$, and the goods this buyer desires must have positive prices. 
Now, the conditions of Lemma \ref{lem.infeasible} can be made to hold by setting the prices of goods in $G'$ to zero.
\end{proof}

\subsection{Details of Stage II}
\label{sec.alg-main-II}

Algorithm \ref{alg.main2} gives the pseudo code for Stage II. In this section, we give the subroutines used by this stage
and Section \ref{sec.pred} gives formal definitions of the predicates used in the While loops.
If Stage I terminates with a feasible price vector $\pp$, the algorithm moves to Stage II to find equilibrium prices.
Since $\pp$ is small, Stage II needs to raise prices of goods to get to the equilibrium. Another way to view the situation is that
since buyers have surplus money, Stage II needs to change prices in such a way that the surplus drops to zero. 
Are both these requirements compatible, i.e., will the surplus of buyers decrease by raising prices? The following lemma clarifies 
this crucial point in the simplified setting of Lemma \ref{lem.xf-max}, i.e., prices of all goods are raised;
of course, Stage II will raise the prices of well-chosen subsets of $G$.

\begin{lemma}
\label{lem.both}
If in the setting of Lemma \ref{lem.xf-max}, prices $\pp$ are feasible and $x$ is raised without violating the Invariant,
then the surplus of each buyer decreases and the resulting price vector is still feasible.
\end{lemma}

\begin{proof}
Since $\pp$ is feasible, for each buyer $i$, $\bt_i < 0$. Clearly, if $x > 1, \ x \cdot \bt_i < \bt_i$, i.e. the surplus of buyer $i$
decreases. Moreover, the property that the $\bt$ of each buyer is negative is preserved. Hence the resulting price vector is still feasible.
\end{proof}

A run of Stage II is also partitioned into {\em phases}, which are further partitioned into {\em iterations}. 
An iteration ends when a new edge is added to the network and a phase ends when a new set goes tight. 
We will say that $S \subseteq G$ is a {\em tight set} if the total price of goods in $S$ exactly equals the
money possessed by buyers who are interested in goods in $S$, i.e., $p(S) = m(\Gamma(S))$.
Clearly, if $S$ is tight, buyers in $\Gamma(S)$ must have zero surplus and hence have $\bt_i = -1$.
In each iteration, the algorithm computes a balanced flow in the current network, $N(\pp)$.

The subroutines used in Stage II are:
\begin{itemize}
\item
{\bf Find sets(II):}
Sets $I \subseteq B$ and $J \subseteq G$ are initialized as follows.
\[ I \la \  \arg \max_{i \in B} \{ \theta_i \}   \ \ \mbox{and} \ \ J \la \ \Ga(I) . \]
All edges are removed from goods in $J$ to buyers in $B - I$; this is justified in
Lemma \ref{lem.remove} below.

\item
{\bf Update sets(II):}
Find the set, $I'$, of all buyers in $B - I$ that have residual paths to buyers in $I$. 
Update 
\[  I \la (I \cup I') \ \ \mbox{and} \ \  J \la \Ga(I). \]
All edges are removed from goods in $J$ to buyers in $B - I$. Once again, this
is justified in Lemma \ref{lem.remove}.
\end{itemize}

Observe that if $(j, i)$ is the new edge added to $S_i$, then good $j$ must move from $G - J$ to $J$,
whether or not $I' = \emptyset$.
The choice of set $J$ above ensures that if the prices of goods in $J$ are increased by an infinitesimally 
small amount (by increasing $x$ as stated in Algorithm \ref{alg.main2}),
there is no change in the maximum bang-per-buck goods of buyers in $B$.

\begin{lemma}
\label{lem.neigh2}
In Stage II, at the start of each iteration, for each buyer $i \in (B - I)$ there is a good $j \in (G - J)$ such that
edge $(j, i)$ is in the network.
\end{lemma}

\begin{proof}
Since $\theta_i < 1$, the balanced flow must be sending flow on some edge $(j, i)$.
If $j \in J$, then there will be a residual path from $i$ to a buyer in $I$, violating
Property 1. Therefore, $j \in (G - J)$.
\end{proof}

\begin{lemma}
\label{lem.remove}
In Stage I (Stage II), the Invariant holds after all edges from goods in $G - J$ ($J$) to buyers in $I$ ($B - I$) are removed.
\end{lemma}

\begin{proof}
The idea of the proof is the same for both statements.
In Stage I, right after {\bf Update sets(I)} is executed, there are no residual paths from $I$ to $B - I$. Therefore, by Property 1,
any edges from $G - J$ to $I$ could not be carrying any flow and hence their removal will not affect the Invariant.

In Stage II, right after {\bf Update sets(II)} is executed, there are no residual paths from $B - I$ to $I$. Therefore, by Property 1,
any edges from $J$ to $B - I$ could not be carrying any flow and hence their removal will not affect the Invariant.
\end{proof}

In each iteration, we need to compute the smallest value of $x$ at which a new edge is added to the network or
a new set goes tight. The former computation is the same as in Stage I. 
Let the smallest value of $x$ at which a new set goes tight be $x^*$. Let $b = \min_{i \in I} \left\{ - {1 \over \bt_i } \right\}$.
Clearly, $b > 1$. 
Using Lemma \ref{lem.xf-max}, proved in Section \ref{sec.high}, for $x$ in the range $1 \leq x \leq b$,
we prove below that $x^* =  b$.

\begin{lemma}
\label{lem.x^*}
$x^* =  b$.
\end{lemma}

\begin{proof}
By definition of $b$, for $1 \leq x < b$, for each $i \in I$, the surplus of $i$ will be $1 + x \bt_i > 0$, since $x \bt_i > -1$.
Thus, each edge $(i, t)$ will have positive surplus, implying that there are no tight sets.

Next, assume that $x = b$. Let 
\[ T = \left\{i \in I ~|~ - {1 \over \bt_i} = b \right\}  \ \ \mbox{and} \ \ 
S = \{ j \in J ~|~ f(j, i) > 0, \ \ \mbox{for some} \ \  i \in T \} . \]
Since $f$ is a balanced flow in $N(\pp)$, there cannot be an edge $(j, i)$ for $j \in S$ and
$i \in (I - T)$ in $N(\pp)$. This is so because otherwise there would be a path from $T$ to $i$ 
in the residual graph, contradicting Property 1 (observe that $\bt_i < -1/ b$). Therefore, $\Ga(S) = T$.
Moreover, for $i \in T$, the surplus on edge $(i, t)$ w.r.t. flow $x \cdot f$ in $N(x \pp)$ is $1 + x (-1/b) = 0$.
Hence $S$ is a tight set in network $N(x \pp)$ for $x = b$.
\end{proof}

\subsection{Predicates used in While loops}
\label{sec.pred}

In Step \ref{step.start} (Stage I), ``a proof of feasibility is reached'' when $\forall i \in B, \ \bt_i < 0$, and
``a proof of infeasibility is reached'' when $\sum_{i \in B} {\beta_i} \geq 0$.
Thus ``a proof of feasibility or infeasibility is not reached'' is satisfied iff
$\neg ((\sum_{i \in B} {\beta_i} \geq 0) \vee (\forall i \in B, \ \bt_i < 0))$.

In Step \ref{step.check} (Stage I), ``$B - I$ desire $J$'' is satisfied iff 
$\exists i \in (B - I), \ \exists j \in J: \ \ u_{ij} > 0$,
and ``buyers in $I$ have small surplus'' is satisfied iff $\forall i \in I, \ \bt_i < 0$.

In Step \ref{step.phase2}, ``a buyer in $B$ has surplus money'' is satisfied iff $\exists i \in B, \ \theta_i > 0)$.

In Step \ref{step.tight}, ``no set in $J$ is tight'' is satisfied iff 
$\neg (\exists S, \ \emptyset \subset S \subseteq J \ s.t. \ S \ \mbox{is tight})$.

\subsection{The role of balanced flow}
\label{sec.use}

Besides being used for defining the central notion of feasible prices,
balanced flow plays the following three, rather diverse, crucial roles in both stages of our algorithm. 

\begin{enumerate}
\item 
Ensure that edges, that need to be removed as prices of goods in $J$ are raised, did not carry any flow and hence their removal
would not violate the Invariant; this is argued in Lemma \ref{lem.remove}

\item
Ensure that in each iteration, buyers entering $I$  in Stage I (Stage II) have sufficiently large $|\bt_i|$ ($|\theta_i|$);
this is established in Lemma \ref{lem.delta-two1} (Lemma \ref{lem.delta-two2}).
\item
Prove that sufficient progress is made in an iteration and hence in a phase. This is established in Lemma \ref{lem.delta-three1}
for Stage I and Lemma \ref{lem.delta-three2} for Stage II. 
\end{enumerate}

As stated in the Introduction,
balanced flow could have been defined without resorting to the $\ltwo$ norm -- as a max-flow that makes the surplus vector lexicographically
smallest, after its components are sorted in decreasing order (and hence making the components as balanced as possible). It is easy to prove
Property 1 with this definition as well. The first two roles listed above make use of Property 1 only. On the other hand, the third
role uses the definition of balanced flow via the $\ltwo$ norm and as argued in Section \ref{sec.tight}, the use of the 
$\ltwo$ norm seems indispensable.



\noindent

\fbox{
\begin{algorithm}{\label{alg.main} (Initialization and Stage I of the Algorithm for ADNB)}

\step
Initialization: 
\begin{description}
\item [(i)]
$\forall i \in B: \ \ m_i \la 1$.

\item [(ii)]
Use the DPSV algorithm to compute equilibrium prices, $\pp$.

\item [(iii)]
$\forall i \in B: \ \ m_i \la 1 + {c_i \over \ga_i}$.

\item [(iv)]
$B_c \la B;  \ \ \ \ \ \  G_c \la G$.

\item [(v)]
$B' \la \emptyset; \ \ \ \ \ \  G' \la \emptyset$.

\item [(vi)]
Compute a balanced flow in $N(\pp)$. 
\end{description}

\bigskip

\begin{center}
{\bf Stage I}
\end{center}

\step
\label{step.start}
{\em (New Phase)} {\bf While} \ a proof of feasibility or infeasibility is not reached \  {\bf do:}

\step
{\bf Find sets(I)}.

\step 
\label{step.check}
{\em (New Iteration)}  {\bf While} \ $B_c - I$ desire $J$ and buyers in $I$ have small surplus \ {\bf do:}

\step
\begin{description}
\item 
Multiply the prices of goods in $J$ and $\alpha$'s of buyers in $I$ by $x$. 
\item
Initialize $x \la 1$, and decrease $x$ continuously until: 
\item
A new edge $(j, i)$ enters $S_i$, for $j \in J$ and $i \in (B_c - I)$.\\ 
Add $(j, i)$ to $N(\pp)$ and compute a balanced flow in it.\\
{\bf Update sets(I)}.
\end{description}

\step {\bf End} {\em (End Iteration)}

\step
\label{step.zero}
If $\forall i \in (B_c - I), \ \forall j \in J: \ \ u_{ij} = 0$,  then:

\begin{description}
\item
Declare $I$ and $J$ adaptable, \ i.e,,
\item
$G' \la (G' \cup J)$ \ \  and \ \ $G_c \la (G_c - J) $. 
\item
$B' \la (B' \cup I)$ \ \ and \ \ $B_c \la (B_c - I)$.
\end{description}

\step {\bf End} {\em (End Phase)}

\step
\label{step.restore}
If $\forall i \in B_c, \ \bt_i < 0$, then:
\begin{description}
\item
Restore prices of adaptable goods in $G'$.
\item
Compute a balanced flow in $N(\pp)$.
\item
Go to Step \ref{step.phase2} in Stage II. 
\end{description}

\step
\label{step.output}
Else (i.e., $\sum_{i \in B_c} {\beta_i} \geq 0$), output ``The game is infeasible''.\\
HALT.

\end{algorithm}
}

\bigskip


\noindent

\fbox{
\begin{algorithm}{\label{alg.main2} (Stage II of the Algorithm for ADNB)}

\step
\label{step.phase2}
{\em (New Phase)} {\bf While} \ a buyer in $B$ has surplus money \ {\bf do:}

\step
{\bf Find sets(II)}.

\step
\label{step.tight}
{\em (New Iteration)} {\bf While} \ no set in $J$ is tight \  {\bf do:}

\step
\begin{description}
\item
Multiply prices of goods in $J$ and $\alpha$'s of buyers in $I$ by $x$. 
\item
Initialize $x \la 1$, and raise $x$ continuously until:
\item
A new edge $(j, i)$ enters $S_i$, for $j \in (G-J)$ and $i \in I$. \\
If so, add $(j, i)$ to $N(\pp)$ and compute a balanced flow in it. \\
{\bf Update sets(II)}.
\end{description}

\step {\bf End} {\em (End Iteration)}

\step {\bf End} {\em (End Phase)}

\step
\label{step.output2}
Output the current allocations and prices. \\
HALT.

\end{algorithm}
}



\section{Running Time Analysis}
\label{sec.time}

We first define some parameters of the given problem instance. Recall that $g = |G|$ and $n = |B|$.
Let $U = \max_{i \in B, j \in G} \{u_{ij} \}$, $C = \max_{i \in B} {c_i}$, and $\Delta = n C U^n$. 
Observe that program (\ref{CP-ADNB}) with all $c_i = 0$ is the same as the convex program for a linear Fisher market with all 
buyers having unit money. Hence, Theorem \ref{thm.bits} gives a lower bound on the price of a good computed in Initialization.
Let this lower bound be $1 / \mu, \ \mu \in \Zplus$.

The following enhanced version of Lemma \ref{lem.xf-max} will be needed in both stages.

\begin{lemma}
\label{lem.xf-balanced}
Let $f$ be a balanced flow in network $N(\pp)$. Then,
for $0 < x \leq b$, the flow $x \cdot f$ is a balanced flow in $N(x \pp)$. 
\end{lemma}

\begin{proof}
For $i, j \in B$ assume that $1 + x \bt_i < 1 + x \bt_j$. Since $x > 0$, $1 + \bt_i < 1 + \bt_j$, i.e.,
w.r.t. flow $f$ in $N(\pp)$, the surplus of $i$ is smaller than that of $j$.
Since $f$ is a balanced flow in $N(\pp)$, by Property 1, there is no path from $i$ to $j$ in the residual
graph. Therefore, w.r.t. flow $x \cdot f$ in $N(x \pp)$ also there is no path from $i$ to $j$ in the residual
graph. Therefore, flow $x \cdot f$ in $N(x \pp)$ satisfies Property 1 and hence is a balanced flow.
\end{proof}

\subsection{Stage I}
\label{sec.time1}

Throughout Stage I, we will consider a partitioning of $B_c$ into two sets, $B_1$ and $B_2$, containing 
buyers having $\bt_i < 0$ and $\bt_i \geq 0$, respectively. 
For Stage I, we will work with the following potential function:
\[ \Phi = \sum_{i \in B_1} {\bt_i^2}  . \]

As Stage I proceeds, buyers move from $B_c$ to $B'$, and within $B_c$ between the sets $B_1$ and $B_2$. For this reason, it will
be convenient to define $\Phi$ using an $n$-dimensional vector, $\psi$, called the {\em associated vector} of network $N$.
The $i$-th component of this vector,
$\psi_i$, is $\bt_i$ for $i \in B_1$, and is 0 for $i \in (B' \cup B_2)$. Hence, an alternative definition for the potential function is: 
\[ \Phi = \| \psi \|^2 . \]

\begin{lemma}
\label{lem.iterations1}
In Stage I, a phase consists of at most $ng$ iterations.
\end{lemma}

\begin{proof}
Observe that if in {\bf Update sets(I)}, $I' = \emptyset$, then a good must move from $J$ to $G_c - J$.
Otherwise, a buyer must move from $B_c - I$ to $I$. Clearly, there can be at most $|J| < |G_c|$ contiguous iterations
of the first type and a total of at most $|B_c - I_0| < |B_c|$ iterations of the second type, where
$I_0$ is the set $I$ at the start of the phase.
\end{proof}

The central fact established below is that $\Phi$ drops by a factor of $(1 - 1/(gn^2))$ in a 
phase (Lemma \ref{lem.delta1}). Towards this end, assume that a given phase consists of $k$ iterations.
Let $I_0$ denote set $I$ at the start of the phase and
let $I_l$ denote the set $I$ at the end of the $l$-th iteration, $1 \leq l \leq k$.
Assume that at the start of this phase, $\max_{i \in B_1} \{ |\bt_i| \} = \dt = \dt_0$. Let
\[ \dt_l = \min_{i \in I_l} \{ |\bt_i| \}, \ \ \mbox{for} \ \  1 \leq l < k, \ \ \mbox{and} \ \ \dt_{k} = 0. \]

As we will see in this section, the potential function $\Phi$ drops monotonically in each iteration in the phase.
Within an iteration, we will account for the drop in two steps. First, as prices of goods in $J$ are reduced, 
by Lemma \ref{lem.xf-max} the $\bt_i$'s of buyers $i \in I$ increase, leading to a reduction in $\Phi$. Second, when a 
new edge $(j, i)$, with $j \in J$ and $i \in (B_c - I)$, is added to the network, the flow becomes more balanced, leading to a further drop.
We will account for these two reductions separately, via different arguments (see Lemma \ref{lem.delta-three1}).
For the first step, we work with the $\lone$ norm, establishing an increase in 
$\sum_{i \in B_1} {\bt_i}$. In the second step, $\sum_{i \in B_1} {\bt_i}$ will not change if $i \in (B_1 - I)$. 
Instead, we establish a decrease in $\|\psi\|^2$ using an $\ltwo$ norm
based argument. We observe that the latter argument is difficult to apply to the first
step since the money of buyers changes as prices change. Also, we do not know of a simple one 
step argument that accounts for the entire reduction in an iteration.

Next, we prove a key fact that accounts for the second decrease.
Just before new edge $(j, i)$ is added to $S_i$, let $N$ be the network and $\pp$ be the prices of goods.
Let $N'$ be the network obtained by adding this edge to $N$; of course, the prices remain unchanged.
Let $f$ and $f^*$ be balanced flows in $N$ and $N'$, respectively, and let
$\psi_i$ and $\psi_i^*$ be their associated vectors.

\begin{lemma}
\label{lem.delta-one1}
$\|\psi\|^2 - \|\psi^*\|^2  \geq  \sum_{h \in B_1} {(\psi_h - \psi_h^*)^2}$. 
\end{lemma}

\begin{proof}
Since the Invariant holds and the prices are unchanged, $f$ and $f^*$ have the same value.
Therefore, flow $f^* - f$ will consist of circulations. Since $f$ is a balanced flow, all these circulations must
use the edge $(j, i)$, because otherwise a circulation not using edge $(j, i)$ could be used for making $f$ more balanced. 
These circulations will have the effect of increasing the surplus of certain buyers in
$I$, say $i_l$, for $1 \leq l \leq k$, and decreasing the surplus of buyer $i \in (B_c - I)$. 
Let $\bt_i - \bt_i^* = \delta$, and for $1 \leq l \leq k$, $\bt_{i_l}^* - \bt_{i_l} = \delta_l$. 
Then, $\sum_{l=1}^k{\delta_l} = \delta$. 

For each buyer $i_l, \ 1 \leq i_l \leq k$, there is a path from $i_l$ to $i$ in the corresponding
circulation and hence there is a path from $i$ to $i_l$ in the residual graph w.r.t. flow $f^*$.
Since $f^*$ is balanced, by Property 1, the surplus of buyer $i$ is at least as large as that of $i_l$. 
Therefore, $\bt_i^* \geq \bt_{i_l}^*$. 

In going from $N$ to $N'$, the $\psi_h$ values can change only for $h = i$, and $h = i_l$, for $1 \leq l \leq k$.
We will consider 3 cases.

{\bf Case 1:} $\bt_i \geq 0$, i.e., $i \in B_2$, and $\bt_i^* \geq 0$. \\
In this case, $\psi_i = \psi_i^* = 0$ and the lemma is obvious.

{\bf Case 2:} $\bt_i \geq 0$, i.e., $i \in B_2$, and $\bt_i^* < 0$. \\
Let $a = -\bt_i^*$. In this case, $\psi_i = 0$ and $\psi_i^* = -a$.

Clearly, $a \leq \sum_{l=1}^k{\delta_l}$
Since $\bt_i^* \geq \bt_{i_l}^*$, $a \leq b_l - \delta_l$.
Now, 
\[  \|\psi\|^2 - \|\psi^*\|^2 \ = \ \left( 0^2 + \sum_{l=1}^k{b_l^2} \right)  - \left( a^2 + \sum_{l=1}^k{(b_l - \delta_l)^2} \right) 
= -a^2 + \sum_{l=1}^k{(2 b_l - \delta_l) \delta_l} \]
\[ \geq -a^2 + \sum_{l=1}^k{(2a + \delta_l) \delta_l}  \geq  -a^2 + \sum_{l=1}^k{\delta_l^2} +  2a \sum_{l=1}^k{\delta_l} 
\geq \sum_{l=1}^k{\delta_l^2} ,\] 
where the first inequality follows from $a \leq b_l - \delta_l$ and the third one follows from $a \leq \sum_{l=1}^k{\delta_l}$.

{\bf Case 3:}  $\bt_i^* < 0$, i.e., $i \in (B_1 - I)$. \\
Clearly, in this case, $\bt_i < 0$.
Substitute $a = -\bt_i$, and for $1 \leq i_l \leq k$,
substitute $b_l = - \bt_{i_l}$. Now, by Lemma \ref{lem.shiftflow1}, to get $\|\psi\|^2 - \|\psi^*\|^2  \geq  \delta^2$.
Clearly, $\delta^2 \geq \sum_{l=1}^k{\delta_l^2}$, giving the lemma.
\end{proof}

\begin{lemma}
\label{lem.shiftflow1}
Let $\delta,\delta_l \geq 0, \ l=1,2,\ldots,k$, with $\delta = \sum_{l=1}^k \delta_l$.
If $a + \dt \leq b_l - \dt_l$, for $l=1,2,\ldots,k$ then
\[  \|{(a,b_1,b_2,\ldots,b_k)}\|^2 -
\|{(a+\delta,b_1-\delta_1,b_2-\delta_2, \ldots,b_k-\delta_k)}\|^2  \geq \delta^2 .\]
\end{lemma}

\begin{proof}
\[ \left( a^2 + \sum_{l=1}^k {b_l}^2 \right) 
 - \left((a+\delta)^2 + \sum_{l=1}^k {(b_l -\delta_l)^2} \right) \]

\[ \geq  \left( (a + \dt - \dt)^2 + \sum_{l=1}^k (b_l - \dt_l + \dt_l)^2 \right) 
 - \left((a+\delta)^2 + \sum_{l=1}^k {(b_l -\delta_l)^2} \right) \]

\[ \geq \dt^2 + 2(a + \dt) \left( (\sum_{l=1}^k {\dt_l}) - \dt \right) \ \ \geq \ \ \dt^2 .\]
\end{proof}

Let $\psi^0$ denote vector $\psi$ at the start of the phase and $\psi^l$ denote $\psi$ at the 
end of iteration $l$, for $1 \leq l \leq k$.

\begin{lemma}
\label{lem.delta-two1}
In the $l$-th iteration, there is a buyer $i \in I_{l-1}$ such that
$|\bt_i|$ decreases by at least $(\delta_{l-1} - \dt_l)$, for $1 \leq l < k$.
\end{lemma}

\begin{proof}
By the definition of set $I'$ in procedure {\bf Update sets(I)} and Property 1, there is a 
buyer $i \in I_{l-1}$ which achieves \ $\min_{i \in I_l} \{ |\bt_i | \}$ \ at the end of iteration $l$.
Clearly, $\bt_i$ increases (and hence $|\bt_i|$ decreases) by at least $(\delta_{l-1} - \dt_l)$ in the $l$-th
iteration.
\end{proof}

\begin{lemma}
\label{lem.delta-three1}
For $1 \leq l \leq k$, 
\[  \|\psi^{l-1}\|^2 - \|\psi^l\|^2  \geq  (\delta_{l-1} - \dt_l)^2 . \]
\end{lemma}

\begin{proof}
We first prove the statement for $1 \leq l < k$.
By Lemma \ref{lem.delta-two1}, there is a buyer $i \in I_{l-1}$ such that
$\bt_i$ increases by at least $(\delta_{l-1} - \dt_l)$ in the $l$-th iteration.
Let us split this increase into two parts, the increase due to decrease in the prices
of goods in $J$ and that due to a new edge entering the network. Let these be
$a$ and $b$, respectively. Therefore, $a + b \ = \ \delta_{l-1} - \dt_l$.

Let $\psi'$ be the vector $\psi$ just before the new edge is added to the network in iteration $l$,
i.e., right after all the decrease in prices of $J$ has happened. As prices in $J$ decrease, the 
beta's of buyers in $I$ increase, each leading to a decrease in $\|\psi'\|^2$; 
clearly, the beta's of buyers in $B_c - I$ remain unchanged.
Let $c$ be the value of beta of buyer $i$ at the beginning of iteration $l$. Then,
\[ \|\psi^{l-1}\|^2 - \|\psi'\|^2  \geq  c^2 - (c + a)^2 = a^2 -   2 a c .\]
By Lemma \ref{lem.delta-one1}, 
\[  \|\psi'\|^2 - \|\psi^l\|^2  \geq  b^2 . \]
Adding the two we get
\[  \|\psi^{l-1}\|^2 - \|\psi^l\|^2  \geq   a^2 -   2 a c + b^2  \geq (\dt' + \dt'')^2 
\geq  (\delta_{l-1} - \dt_l)^2 , \]
where the second last inequality follows from the observation that $b \leq -c$.

Finally, in the $k$-th iteration, there is a buyer $i \in I_{k-1}$ whose $\psi_i$ changes
from $\beta_i < 0$ to 0. Therefore, 
\[  \|\psi^{k-1}\|^2 - \|\psi^k\|^2  \geq  \bt_i^2 \geq  (\delta_{k-1} - \dt_k)^2 , \]
since $\delta_{k-1} \leq - \bt_i$ and $\dt_k = 0$.
\end{proof}

\begin{lemma}
\label{lem.delta1}
In a phase in Stage I, the potential drops by a factor of 
\[\left(1 - {1 \over {n^2 g}} \right) . \]
\end{lemma}

\begin{proof}
Now, $\|\psi^{0}\|^2 - \|\psi^k\|^2$
can be written as a telescoping sum of $k$ terms, each of which is the decrease in the potential in one of the
$k$ iterations. Lemma \ref{lem.delta-three1} gives a lower bound on each of these terms.
The total lower bound is minimized when each of the differences $(\delta_{l-1} - \dt_l)$ is equal.
Now using the fact that $\delta_0 = \delta$ and $\delta_k = 0$, we get:
\[  \|\psi^{0}\|^2 - \|\psi^k\|^2 \geq {\dt^2 \over k} . \]
Finally, since $\|\psi^{0}\|^2 \leq n \dt^2$, and by Lemma \ref{lem.iterations1} $k \leq ng$,  we get:
\[  \|\psi^{k}\|^2 \leq  \|\psi^0\|^2 \left(1 - {1 \over {n^2 g}} \right) . \]
\end{proof}

\begin{lemma}
\label{lem.lb1}
At any point in Stage I, if $\sum_{i \in B_2} {\bt_i} > 0$, then 
\[ \sum_{i \in B_2} {\bt_i} \geq  {1 \over {U^n \mu^g}} .\]
\end{lemma}

\begin{proof}
First observe that the maximum bang-per-buck of buyers in $B_2$ remains unchanged throughout
Stage I, and is determined by prices of goods found in the Initialization. 
Let $J_2 = \Ga(B_2)$. By Property 1, there is no flow from a good in $J_2$ to a buyer in $B_1$,
and therefore, all flow from $J_2$ must go to buyers in $B_2$. Therefore,  
\[  \sum_{i \in B_2} {\bt_i} = \sum_{i \in B_2} {\theta_i} - |B_2|  
= \sum_{i \in B_2} {\alpha_i} - \sum_{j \in J_2} {p_j} . \]
Now if $\sum_{i \in B_2} {\bt_i} > 0$, then the denominator of this sum is a product of
at most $n$ $u_{ij}$'s and at most $g$ $p_j$'s, and is therefore bounded by $U^n \mu^g$,
proving the lemma.
\end{proof}

\begin{lemma}
\label{lem.StageI}
The execution of Stage I requires at most 
\[ O\left(n^4 g^2 (n\log U + g \log \mu) \right) \]
max-flow computations.
\end{lemma}

\begin{proof}
By Lemma \ref{lem.delta1}, the square of the potential drops by a factor of two after
$O(n^2 g)$ phases. At the start of the algorithm, the potential is at most $n$. 

If at any point, $\sum_{i \in B_2} {\bt_i} = 0$, i.e., $\forall i \in B_2: \ \bt_i = 0$, 
each subsequent phase must end up in Step 5 and not Step 6, i.e., some buyers will be removed from consideration.
Therefore, from this point on, at most $n$ more phases are needed for Stage I to terminate. 

Next, let's assume that  $\sum_{i \in B_2} {\bt_i} > 0$ throughout Stage I. If so, by 
Lemma \ref{lem.lb1}, once the potential drops below ${1 \over {(U^n \mu^g)}^2}$, the phase must end
(in Step 7). Therefore the number of phases is 
\[ O(n^2 g \log (U^{2n} \mu^{2g}) ) .\]
By Lemma \ref{lem.iterations1} each phase consists of at most $n g$ iterations and 
each iteration requires $n$ max-flow computations for finding a balanced flow.
The lemma follows.
\end{proof}


\subsection{Stage II}
\label{sec.time2}

When $\forall i \in B: \ \bt_i <0$, the algorithm starts with Stage II.
Since in this stage the algorithm only raises prices of goods (i.e., increases $x$),
by Lemma \ref{lem.both}, $\forall i \in B: \ \bt_i <0$ holds until termination.

In this section, we will work with the $\theta_i$'s of buyers, rather than their $\bt_i$'s.
Thus, throughout Stage II, $\forall i \in B: \ \theta_i <1$. 
For Stage II, we will work with the following potential function:
\[ \Phi = \sum_{i \in B} {\theta_i^2}  . \]

\begin{lemma}
\label{lem.denom}
In Stage II, at the termination of a phase, the prices of goods in the newly tight set
must be rational numbers with denominator $\leq \Delta$.
\end{lemma}

\begin{proof}
Let $S$ be the newly tight set. Consider the subgraph of the network induced on the 
bipartition $(S, \Gamma(S))$, and view this as an undirected graph, say $H$.
Assume w.l.o.g. that this graph is connected (otherwise we prove the lemma for
each connected component of $H$). Pick a spanning tree in $H$. 

Pick any good $j \in S$, and find a path in the spanning tree from $j$ to each good $j' \in S$.
If $j$ reaches $j'$ with a path of length $2 l$, then $p_{j'} = a p_j /b$ where $a$ and $b$ are 
products of $l$ utility parameters ($u_{ik}$'s) each. Since alternate edges of this path contribute to
$a$ and $b$, we can partition the $u_{ik}$'s of edges in the spanning tree into two sets, $T_1$ and $T_2$, 
such that $a$ uses $u_{ik}$'s from $T_1$ and $b$ uses those from $T_2$.

Next, consider $\alpha_i = c_i/\ga_i$, for $i \in \Ga(S)$. Now, $\ga_i = u_{i, j'} / p_{j'}$, where
$(i, j')$ is any edge in the network. Find the path in the spanning tree from $i$ to $j$ and use the 
first edge on this path for computing $\ga_i$ (it is easy to see that all these edges come from 
set $T_2$), and substitute $p_{j'}$ using the expression stated above, i.e., $p_{j'} = a p_j /b$. 

Since $S$ is a tight set, 
\[  \sum_{j' \in S} {p_{j'}}  =  \sum_{i \in \Ga(S)} {1 + \alpha_i} . \]
In this equation, substitute for $p_{j'}$ and $\alpha_i$ using the expressions constructed above
to get an equation with one variable, i.e., $p_j$. Now, it is easy to see that the denominator of $p_j$
is $ \leq \Delta$. 
\end{proof}

\begin{lemma}
\label{lem.phases}
In Stage II, consider two phases $P$ and $P'$, not necessarily consecutive,
such that good $j$ lies in the newly tight sets at the end of $P$ as well as $P'$.
Then the increase in the price of $j$, going from $P$ to $P'$, is at least $1/\Delta^2$.
\end{lemma}

\begin{proof}
Let the prices of $j$ at the end of $P$ and $P'$ be $p/q$ and $r/s$,
respectively. Clearly, $r/s > p/q$.
By Lemma \ref{lem.denom}, $q \leq \Delta$ and $r \leq \Delta$. Therefore
the increase in price of $j$,
\[ {r \over s} - {p \over q} \geq {1 \over \Delta^2} . \]
\end{proof}

\begin{lemma}
\label{lem.it2}
In Stage II, a phase consists of at most $g$ iterations.
\end{lemma}

\begin{proof}
After each iteration, other than the last one, at least one good will move from $G - J$ to $J$.
\end{proof}

The structure of the rest of the argument is quite similar to that of Stage I.
Once again, the central fact established is that $\Phi$ drops by an inverse polynomial factor, of $(1 - 1/n^2)$, in a 
phase (Lemma \ref{lem.delta2}). Assume that a given phase consists of $k$ iterations.
Let $I_0$ denote set $I$ at the start of the phase and
let $I_l$ denote the set $I$ at the end of the $l$-th iteration, $1 \leq l \leq k$.
Assume that at the start of this phase, $\max_{i \in B} \{ |\theta_i| \} = \dt = \dt_0$. Let
\[ \dt_l = \min_{i \in I_l} \{ |\th_i| \}, \ \ \mbox{for} \ \  1 \leq l < k, \ \ \mbox{and} \ \ \th_{k} = 0. \]

As in Stage I, we will account for the drop in $\Phi$ in two steps in each iteration. First, as 
prices of goods in $J$ are increased, the $\th_i$'s of buyers $i \in I$ decrease, leading to a reduction in $\Phi$. 
Second, when a new edge $(j, i)$, with $j \in (G_c - J)$ and $i \in I$, is added to the network, the flow becomes more balanced, 
leading to a further drop. As in Stage I, we will account for the first drop using the $\lone$ norm and
the second drop using the $\ltwo$ norm.

We begin by accounting for the second decrease.
Just before new edge $(j, i)$ is added, let $N$ be the network and $\pp$ be the prices of goods.
Let $N'$ be the network obtained by adding this edge to $N$; of course, the prices remain unchanged.
Let $f$ and $f^*$ be balanced flows in $N$ and $N'$, respectively.
Denote by $\th$ ($\th^*$) the surplus vector w.r.t. flow $f$ in $N$ (flow $f^*$ in $N'$).

\begin{lemma}
\label{lem.delta-one2}
$\|\th\|^2 - \|\th^*\|^2  \geq  \delta^2$, where $\delta = \th_i - \th_i^*$.
\end{lemma}

\begin{proof}
Since the Invariant holds and the prices are unchanged, $f$ and $f^*$ have the same value.
Therefore, flow $f^* - f$ will consist of circulations. Since $f$ is a balanced flow, all these circulations must
use the edge $(j, i)$, because otherwise a circulation not using edge $(j, i)$ could be used for making $f$ more balanced. 
These circulations will have the effect of decreasing the surplus of buyer 
$i \in I$, and increasing the surplus of buyers $i_l \in (B - I)$, for $1 \leq l \leq k$.
Let $\th_{i_l}^* - \th_{i_l} = \delta_l$, for $1 \leq l \leq k$. Then, $\sum_{l=1}^k{\delta_l} = \delta$. 

For each buyer $i_l, \ 1 \leq i_l \leq k$, there is a path from $i_l$ to $i$ in the corresponding
circulation and hence there is a path from $i$ to $i_l$ in the residual graph w.r.t. flow $f^*$.
Since $f^*$ is balanced, by Property 1, the surplus of buyer $i$ w.r.t. $f^*$ is at least as large as that of $i_l$. 
Therefore, $\th_i^* \geq \th_{i_l}^*$. 
The inequality $\|\th\|^2 - \|\th^*\|^2  \geq  \delta^2$ now follows from 
Lemma \ref{lem.shiftflow2}, on substituting $a = \th_i^*$, and for $1 \leq l \leq k$, $b_l = \th_{i_l}^*$.
\end{proof}

\begin{lemma} 
\label{lem.shiftflow2}
If $a \geq b_l\geq 0, l=1,2,\ldots,k$ and $ \delta = \sum_{l=1}^k \delta_l$ where $\delta,\delta_l \geq 0,
l=1,2,\ldots,k$, then 
\[  \|{(a+\delta,b_1-\delta_1,b_2-\delta_2, \ldots,b_k-\delta_k)}\|^2 -
 \|{(a,b_1,b_2,\ldots,b_k)}\|^2   \geq  \delta^2  . \]
\end{lemma}

\begin{proof}
$$(a+\delta)^2 + \sum_{i=1}^k {(b_i -\delta_i)^2} -   a^2 - \sum_{i=1}^k
{b_i}^2  \geq \delta^2 + 2 a
(\delta - \sum_{i=1}^{k}{\delta_i}) \geq \delta^2. $$
\end{proof}

Let $\th^0$ denote the surplus vector at the start of the phase and let $\th^l$ denote the 
surplus vector at the end of iteration $l$, for $1 \leq l \leq k$.

\begin{lemma}
\label{lem.delta-two2}
In the $l$-th iteration, there is a buyer $i \in I_{l-1}$ whose
surplus decreases by at least $(\delta_{l-1} - \dt_l)$, for $1 \leq l < k$.
\end{lemma}

\begin{proof}
By the definition of set $I'$ in procedure {\bf Update sets(II)} and Property 1, there is a 
buyer $i \in I_{l-1}$ which achieves \ $\min_{i \in I_l} \{ |\th_i| \}$ \ at the end of iteration $l$.
Clearly, the surplus of $i$ decreases by at least $(\delta_{l-1} - \dt_l)$ in the $l$-th
iteration.
\end{proof}

\begin{lemma}
\label{lem.delta-three2}
For $1 \leq l \leq k$, 
\[  \|\th^{l-1}\|^2 - \|\th^l\|^2  \geq  (\delta_{l-1} - \dt_l)^2 . \]
\end{lemma}

\begin{proof}
We first prove the statement for $1 \leq l < k$.
By Lemma \ref{lem.delta-two2}, there is a buyer $i \in I_{l-1}$ whose
surplus decreases by at least $(\delta_{l-1} - \dt_l)$ in the $l$-th iteration.
Let us split this decrease into two parts, the decrease due to increase in the prices
of goods in $J$ and that due to a new edge entering the network. Let these be
$a$ and $b$, respectively. Clearly, $a + b \ \geq \ \delta_{l-1} - \dt_l$.

Let $\th'$ be the surplus vector just before the new edge is added to the network in iteration $l$,
i.e., right after all the increase in prices of $J$ has happened. As prices in $J$ increase, the 
surpluses of buyers in $I$ decrease, but those of buyers in $B - I$ remain unchanged.
Let $c$ be the surplus of buyer $i$ at the beginning of iteration $l$. Then,
\[ \|\th^{l-1}\|^2 - \|\th'\|^2  \geq  c^2 - (c - a)^2 \geq a^2  +  2ac .\]
By Lemma \ref{lem.delta-one2}, 
\[  \|\th'\|^2 - \|\th^l\|^2  \geq  b^2 . \]
Adding the two we get
\[  \|\th^{l-1}\|^2 - \|\th^l\|^2  \geq a^2  +  2ac + b^2 \geq (a + b)^2
\geq  (\delta_{l-1} - \dt_l)^2 , \]
where the second last inequality follows from the observation that $b \leq c$.

Finally, in the $k$-th iteration, there is a buyer $i \in I_{k-1}$ whose surplus changes
from $\th_i > 0$ to 0. Therefore, 
\[  \|\th^{k-1}\|^2 - \|\th^k\|^2  \geq  \th_i^2 \geq  (\delta_{k-1} - \dt_k)^2 , \]
since $\delta_{k-1} \leq  \th_i$ and $\dt_k = 0$.
\end{proof}

\begin{lemma}
\label{lem.delta2}
In a phase in Stage II, the potential drops by a factor of 
\[\left(1 - {1 \over {n^2}} \right) . \]
\end{lemma}

\begin{proof}
Now, $\|\th^{0}\|^2 - \|\th^k\|^2$
can be written as a telescoping sum of $k$ terms, each of which is the decrease in the potential in one of the
$k$ iterations. Lemma \ref{lem.delta-three2} gives a lower bound on each of these terms.
The total lower bound is minimized when each of the differences $(\delta_{l-1} - \dt_l)$ is equal.
Now using the fact that $\delta_0 = \delta$ and $\delta_k = 0$, we get:
\[  \|\th^{0}\|^2 - \|\th^k\|^2 \geq {\dt^2 \over k} . \]
Finally, since $\|\th^{0}\|^2 \leq n \dt^2$, and by Lemma \ref{lem.it2} $k \leq n$,  we get:
\[  \|\th^{k}\|^2 \leq  \|\th^0\|^2 \left(1 - {1 \over {n^2}} \right) . \]
\end{proof}

\begin{lemma}
\label{lem.StageII}
The execution of Stage II requires at most
\[ O\left(n^4(\log n + n\log U + \log C) \right) \]
max-flow computations.
\end{lemma}

\begin{proof}
By Lemma \ref{lem.delta2}, the potential drops by a factor of half after
$O(n^2)$ phases. At the start of the algorithm, the potential is at most $n$.
Once its value drops below $1/\Delta^4$, the algorithm requires at most $n$ more phases to 
compute equilibrium prices. This follows from Lemma \ref{lem.denom} and Lemma \ref{lem.phases}.
Therefore the number of phases is 
\[ O(n^2 \log (\Delta^4 n)) \  = \  O(n^2(\log n + {n \log U}  +  \log C )) .\]
By Lemma \ref{lem.it2} each phase consists of $n$ iterations and 
each iteration requires $n$ max-flow computations for computing a balanced flow.
The lemma follows.
\end{proof}

Lemmas \ref{lem.StageI} and \ref{lem.StageII} give:

\begin{theorem}
\label{thm.poly}
Algorithms \ref{alg.main} and \ref{alg.main2} solve the decision and search versions of 
Nash bargaining game \RNB\ using
\[ O\left(n^4 g(\log n + n\log U + \log C  + g \log \mu) \right) \]
max-flow computations.
\end{theorem}

\section{Postmortem and Was Stage I Really Needed?}
\label{sec.postmortem}

As pointed out in Section 4 in \cite{DPSV}, the primal-dual paradigm operates in a fundamentally different way in
the setting of a rational convex program than in the setting of an integral linear program. In the latter setting,
in each iteration, it picks an unsatisfied complementary slackness condition and satisfies it. On the other hand,
in the former setting, the algorithm starts off with a suboptimal solution that can be viewed as relaxing a class 
of the KKT conditions. It then tightens these conditions gradually; when they are all fully tightened, the optimal solution 
has been reached. 

Let us first show that this high level picture applies to Stage II of our algorithm as well. 
Consider the situation right after a balanced flow has been computed at any point in Stage II, and consider an arbitrary buyer $i$.
At this point, let $f_i$ be the flow sent on
edge $(i, t)$ by the balanced flow; $f_i$ is also the money spent by $i$ in the current allocation.
Buyer $i$'s available money at this point is $m_i = 1 + c_i/\ga_i$. Therefore,
\[ f_i = m_i - \bt_i -1 = {c_i \over \ga_i} - \bt_i .\]
Let $v_i$ be the total utility derived by $i$ from the current allocation and suppose that $x_{ij} > 0$. Then,
\[ \ga_i = {u_{ij} \over p_j} = {v_i \over f_i}  =  {v_i \over {(- \bt_i) + {c_i \over \ga_i}}}   .\]
This yields
\[ \left( (- \bt_i) + {c_i \over \ga_i} \right) \ga_i = v_i \ \implies \ \ga_i = {{v_i - c_i} \over {(-\bt_i)}} .\]
Substituting for $\ga_i$ and rearranging we get 
\[ {p_j \over {(- \bt_i)}} = {u_{ij} \over {v_i - c_i}} .\]
To summarize, at any point in Stage II, we have ensured the first two KKT conditions and relaxed the last two as follows: 

\begin{description}
\item  [(1)]
$ \forall j \in G: ~ p_j \geq 0$.

\item [(2)]
$ \forall j \in G: ~ p_j > 0  \  \implies \  \sum_{i\in B} {x_{ij}} = 1$.

\item [(3')]
$ \forall i \in B, \ \forall j \in G: \ {p_j \over {-\bt_i}} \geq {u_{ij} \over {v_i - c_i}}$.

\item [(4')]
$ \forall i \in B, \ \forall j \in G:  x_{ij} > 0 \ \implies {p_j \over {-\bt_i}} = {u_{ij} \over {v_i - c_i}}$.
\end{description}

Observe that if prices are not feasible, then for some $i, \ \bt_i \geq 0$. If so, the relaxed KKT conditions (3') and (4') will be meaningless.
Recall that throughout Stage II, $0 < -\bt_i \leq 1$ and $(-\bt_i)$ monotonically increases and reaches 1 at termination.
Thus, at termination, the last two KKT conditions are also ensured.
For establishing a bound on the number of phases needed in Stage II, it suffices to study the potential function
\[ \Phi' = \sum_{i \in B} {\bt_i^2} .\]
Clearly, $\Phi' > 0$ at the start of Stage II and increases monotonically to $n$. However, it turns out to be more convenient to 
study the potential function $\Phi$ given in Section \ref{sec.time2}, which clearly achieves the same end. 

Stage I determines feasibility of the given market. What if we only wanted to solve the promise problem
of finding the equilibrium of a given feasible market? Since the prices found by Initialization are guaranteed to be small,
could we not go directly to Stage II and raise prices until equilibrium is attained? The answer is ``No''.
The reason is that Stage II is guaranteed to converge only if it is started with feasible prices. 
As can be seen from the proof of Lemma \ref{lem.xf-max}, if for some buyer $i$, $\bt_i > 0$, then raising $x$ will
actually increase her surplus. This happens because her money will increase at a faster rate than the rate at which flow on edge
$(i, t)$ increases. 

Thus, Stage I not only determines feasibility but, if the given market is feasible, it also finds a suitable initial price vector for Stage II. 
An interesting aspect of Stage I is that it needs to decrease prices of goods, even though it starts with a small price vector.
Furthermore, it is easy to see that if the given market is feasible and Stage I is started off with {\em any} small price vector, not necessarily
the one found by Initialization, it will terminate with a feasible price vector. Hence we get the following interesting fact:

\begin{lemma}
\label{lem.interesting}
Let $\CM$ be a feasible flexible budget market and let $\pp$ be small, positive prices for it. Then, there exist positive prices
$\qq$ such that $\pp$ weakly dominates $\qq$ and prices $\qq$ are feasible.
\end{lemma}

Finally we address the question of whether a separate procedure was needed for testing feasibility, especially in light of the 
fact that such a procedure is not needed when the primal-dual paradigm is used for solving a non-total integral linear program. 
Let us illustrate the latter by comparing, at a high level, the algorithms for the problems of maximum weight matching and
maximum weight perfect matching in bipartite graphs; for full details, see \cite{Bill.book}. 

The only difference in the LP's of these two 
problems is that whereas the former demands at most 1 matched edge incident at each vertex and the latter demands exactly 1 edge.
As a result, in the dual, the vertex variables are constrained to be non-negative in the former and unconstrained in the latter.
This gives rise to an additional complementary slackness condition in the former, i.e., any vertex with a positive dual must be matched.

The algorithm for the former problem attempts to ``repair'' this complementary slackness condition one vertex at a time.
It attempts to find an augmenting path from a violating vertex, say $v$,
by growing a ``Hungarian tree'' rooted at $v$. If another unmatched vertex enters the tree, an augmenting path can be found and 
$v$ is matched off. On the other hand, if the tree becomes maximal without encountering another unmatched vertex, then the algorithm is
able to drive the dual of some vertex in the tree down to zero. If this vertex is $v$, the complementary slackness condition at
$v$ has been repaired. If it is some other vertex, say $u$, then there is an alternating path between $u$ and $v$, which enables the
algorithm to unmatch $u$ and instead match off $v$. This repairs the complementary slackness condition at $v$ without creating a
violation at $u$.

The algorithm for the latter problem is very similar -- it attempts to iteratively match off unmatched vertices. 
It tries to find an augmenting path from an unmatched vertex, say $v$, by growing a ``Hungarian tree'' rooted at $v$. 
If the tree becomes maximal without encountering another unmatched vertex, then the number of 
``outer'' vertices in the tree exceeds the number of ``inner'' vertices by 1. Now, decreasing the duals at outer vertices by $\Delta$
and increasing them at inner vertices by $\Delta$ yields a feasible dual. By letting $\Delta \rightarrow \infty$, we get that the dual LP is unbounded.
This gives a proof of infeasibility of the primal LP. 

Consider the two problems \RNB \ and Fisher's linear case with the money of each agent being unit. Observe
that the latter problem is a special case of the former when the disagreement utilities of all agents are zero, and the latter is total
whereas the former is non-total. Why was the algorithm for the former so much more elaborate than that for the latter (i.e., the DPSV algorithm),
especially in view of the fact that the algorithm for maximum weight perfect matching is not more involved than that for maximum weight matching?

In the case of perfect matching, feasibility was ensured one vertex at a time, and when it could not be ensured, we got a proof of infeasibility
right away. In the case of rational convex programs, the KKT conditions were enforced gradually and these conditions were fully satisfied
only right at the end. So, why can't we simply run Stage II right to the end and then obtain a solution or find out that the given
instance was infeasible? The reason is that to guarantee termination of Stage II, we need to start it off with a feasible price vector, i.e.,
we need to determine feasibility before starting with Stage II. Hence Stage I appears to be essential.

\section{$\lone$-norm Does Not Suffice}
\label{sec.tight}

We give a family of examples showing that the DPSV algorithm, for Fisher's linear case,
may end up making only inverse exponential progress in a phase if the potential function used is 
the $\lone$ norm of the surplus vector.

We will define the example in terms of 2 parameters, $\dt$ and $H$, which will be set at the end.
Assume $B = \{b_0, b_1, \ldots, b_{n-1} , b_n \}$ and $G = \{g_0, g_1, \ldots, g_{n-1}, g_n \}$.
At the start of the phase, the only edges present in the network are
$(g_i, b_i)$, for $0 \leq i \leq n$. The money of the buyers are as follows:
\[ m_0 = 1 + \dt, \ \ \ \mbox{and for} \ \ 1 \leq i \leq n-1, \ m_i = {\dt \over {2^i}}, \ \ 
\mbox{and} \ \ m_n = H + {\dt \over n} . \]
The prices of goods are as follows:
\[ p_0 = 1, \ \ \ \mbox{and for} \ \ 1 \leq i \leq n-1, \ p_i = {\dt \over {2^i}}, \ \ 
\mbox{and} \ \ p_n = H + {\dt \over n} . \]
Hence, at the start of the phase, the surplus of $b_0$ is $\dt$, and that of the rest of the 
buyers is 0.

We will set $\dt = 1$ and $H$ to be a large number, say $n$.
The phase starts with $I = \{b_0\}$ and $J = \{g_0\}$.
Assume that at the end of iteration $i$, edge $(g_{i}, b_{i-1})$ enters the network, and
as a result, $b_i$ enters $I$ and $g_i$ enters $J$, for $1 \leq i \leq n$. 
The increment in price in each iteration is very small -- this is easily
arranged by choosing the right utilities $u_{ij}$'s. 

To keep the description clean, let us assume 
the increments in price are all zero; the numbers can be easily modified by inverse exponential
amounts to yield the desired outcome, even if the prices need to increase in each iteration.
If so, at the end of all this, the surplus of $b_i$ is 
\[ {\dt \over {2^{i+1}}}, \ \ \mbox{for} \ \ 0 \leq i \leq n-1 , \]
and that of $b_n$ is ${\dt \over {2^{n-1}}}$.

Finally, in iteration $n+1$, a very slight increase in $x$ leads to set $\{g_n\}$ going tight.
Observe that the reason for choosing $H$ to be a large number is the ensure that this slight
increase in $x$ will not make a larger set go tight. Observe that $\Ga(\{g_n\} = \{ b_n, b_{n+1} \}$.
Now, the increase in the price of $g_n$ needed for this is ${\dt \over {2^{n-1}}}$. Since $H$ is a 
large number and the increase in $x$ is very small, the total increase in the prices of other goods
is at most a constant factor more. 

In summary, the total increase in the $\lone$ norm of $\pp$ in this phase is an inverse
exponential factor.

\section{Discussion}
\label{sec.discussion}

Our paper provides two new pieces of evidence to show that the notion of balanced flow is basic to the problems tackled in \cite{DPSV}
and the current paper. The first is the result of Section \ref{sec.tight}. Second, observe that the notion of a feasible price vector
was defined using balanced flows. In view of Lemma \ref{lem.bounded} and the remarks made after it, and of
Lemma \ref{lem.interesting}, characterizing the sets of small and 
feasible price vectors for a feasible flexible budget market is an interesting question.

As in the case of the EG-program, the optimal solutions of convex program (\ref{CP-ADNB}) resemble
those of a linear program rather than a nonlinear program. So, we repeat a question raised in
\cite{va.chapter} namely, can the solution to \RNB\  be captured via a linear program? 
We believe the answers to these questions are ``no'' and that establishing this in a suitable
formal framework will provide new insights into the boundary between linear and nonlinear programs.

The most prominent problem for which a rational convex program is known but a combinatorial algorithm is not known is
the linear case of the Arrow-Debreu market model, i.e., the convex program of Jain \cite{JainAD}. We pose the following easier
question: Give a polynomial time algorithm for this problem that uses only an LP solver.
A much more general question along these lines is stated in the Introduction.

We list some more rational convex programs and
leave the problem of finding combinatorial algorithms for them. First, generalize \RNB \ to additively 
separable, piecewise-linear, concave utilities. This problem has a rational convex program, and the question of finding a 
combinatorial algorithm for it becomes even more significant in view of recent results showing 
PPAD-completeness of computing an equilibrium in the Arrow-Debreu model with these utility functions \cite{Chen.plc,ChenTeng,VY.plc}.

In an interesting paper, Kalai \cite{Kalai.nonsymmetric} relaxed Nash's axiom of symmetry and derived the solution concept of 
{\em nonsymmetric bargaining games}. The convex program capturing the solution to the nonsymmetric extension of \RNB \ is also rational; 
moreover, this convex program generalizes the Eisenberg-Gale program and hence captures Fisher's linear case as well.
Despite substantial effort, this problem has not yielded to a combinatorial algorithm. Once it is obtained, one could consider
the common generalization of the last two problems, i.e., nonsymmetric \RNB \ with additively 
separable, piecewise-linear, concave utilities.

On restricting \RNB \ (nonsymmetric 
\RNB ) to zero disagreement utilities we get the problems of computing equilibria for linear Fisher markets with unit (arbitrary) money
among buyers. Of course, both these problems are total. It turns out that a combinatorial algorithm for the unit money case is no easier 
than that for the arbitrary money case. In view of this, the difficulty of obtaining a combinatorial algorithm for nonsymmetric
\RNB \ comes as a surprise and may be substantiating the observation that in the setting of rational convex programs, non-total problems
behave quite differently from total problems.
 
All the rational convex programs mentioned above involve the log function in the objective, together with linear constraints.
Another class of rational convex programs is obtained by having a quadratic objective function and linear constraints.
It will be very interesting to obtain combinatorial polynomial time algorithms for such rational convex programs as well.

The reader can see that the algorithm for \RNB\ exploits a surprisingly rich and clean structure which is, in some ways, 
reminiscent of the majestic structure of matching. In our experience, such structure does not occur in isolation and 
we believe that what we see so far is the tip of an iceberg. This leads to the question:
what does the rest of the iceberg look like? 

One possibility is to seek combinatorial approximation
algorithms for solving specific classes of nonlinear convex programs. 
In this respect, important hints may be obtained from the way the primal-dual paradigm was extended
from solving linear programs exactly to obtaining near-optimal solutions to linear programs
within the area of approximation algorithms. The mechanism involved in all of the latter algorithms was 
that of relaxing complementary slackness conditions, which was first formalized in \cite{WGMV}. We also note
that in the setting of approximation algorithms, the primal-dual paradigm has been successful primarily
for minimization problems. So, our more precise question is, ``Is there a natural way of relaxing the KKT 
conditions to obtain primal-dual (combinatorial) algorithms for near-optimally solving 
interesting classes of (perhaps minimization) nonlinear convex programs?''

\section{Acknowledgments}
I am indebted to Ehud Kalai and Nimrod Megiddo for sharing their insights on the Nash bargaining
game and starting me off on this research. I also wish to thank Matthew Andrews for several valuable discussions on the
application of \RNB \ reported in Section \ref{sec.appl}.

\bibliography{kelly} 
\bibliographystyle{alpha}

\appendix

\section{Solution to ADNB in the Limit} 
\label{sec.limit} 

Assume that we are given an instance of game \RNB\ that is feasible and let $\CM$ be 
the flexible budget market obtained from it.
In this section, we present an algorithm that converges to the equilibrium of $\CM$ in the limit.

Algorithm \ref{alg.limit} will use the DPSV algorithm as a subroutine \cite{DPSV}. When this subroutine 
is called, we assume that the money of each agent is fixed and is specified in the vector $\mm$.

Let $N'(\pp)$ denote the network for the case that the money of agents is fixed and specified by 
vector $\mm$; this network differs from network $N(\pp)$ only in that the capacities of edges
going from buyers to $t$ are specified by $\mm$, rather than being defined as a function of the prices.

\bigskip

\noindent

\fbox{
\begin{algorithm}{\label{alg.limit} (Solution to ADNB in the Limit)}

\step
Initialization: $\forall i \in B: \ \ m_i \la 1$.

\step {\label{step.C1}}
Compute equilibrium prices, $\pp$, for market $(\uu, \mm)$ using the DPSV algorithm.

\step
For each $i \in B$, compute $\ga_i$ w.r.t. prices $\pp$, 
and set $m_i' \la 1 + {c_i \over \ga_i}$.

\step
If $\mm' = \mm$ then output equilibrium allocations and {\bf HALT}. \\
Else, update $\mm$ to  $\mm'$  and go to Step \ref{step.C1}.

\end{algorithm}
}

\bigskip

Let $\pp^*$ and $\mm^*$ be the equilibrium prices and moneys for the flexible budget market $\cal M$, and
let $\pp^{(k)}$ and $\mm^{(k)}$ denote the prices and moneys computed by the algorithm in the $k$-th iteration,
$k \geq 1$.

\begin{lemma}
\label{lem.monotone}
$\pp^{(k)}$ and $\mm^{(k)}$ are monotone increasing and are weakly dominated by 
$\pp^*$ and $\mm^*$, respectively.
\end{lemma}

\begin{proof}
We will use the following 2 facts. First, the DPSV algorithm maintains the following invariant throughout:\\
{\bf Invariant:} W.r.t. current prices, $\pp$, $(s, B \cup G \cup t)$ is a min-cut in network $N'(\pp)$. \\
Second, if $\pp$ are equilibrium prices for money $\mm$ and if $\mm'$ is at least as large as $\mm$ 
in each component, 
then the equilibrium prices for money $\mm'$ cannot be smaller than $\pp$ in any component.

Consider the following induction hypothesis:
\begin{itemize}
\item
the algorithm given above maintains the Invariant throughout.
\item
$\pp^{(k)}$ is monotone increasing (hence, for each agent $i$, $\ga_i$ is monotonically decreasing). 
\item
$\mm^{(k)}$ is monotone increasing.
\end{itemize}
It is easy to carry out this induction simultaneously for all 3 assertions.

Using the first assertion and Lemma \ref{lem.bounded}, $\pp^{(k)}$ is weakly dominated bounded by $\pp^*$. 
Now, using the formula for  money in flexible budget markets, it is easy to see that $\mm^{(k)}$ is weakly dominated by $\mm^*$.
\end{proof}

\begin{theorem}
\label{thm.limit}
Algorithm \ref{alg.limit} converges to the equilibrium prices and moneys of market $\cal M$
in the limit.
\end{theorem}

\begin{proof}
We will use the following fact:
for the linear case of Fisher's model, the analog of Lemma \ref{lem.N} holds, i.e., if $\pp$ are equilibrium
prices for money $\mm$, then in network $N'(\pp)$,
$(s, B \cup G \cup y)$ and $(s \cup B \cup G, t)$ must both
be min-cuts (for a proof, see Lemma 5.2 in \cite{va.chapter}).

Since $\pp^{(k)}$ and $\mm^{(k)}$ are monotone increasing and bounded, they must converge. Let 
$\pp^{(0)}$ and $\mm^{(0)}$ be their limit points. W.r.t. these prices and moneys, it must be the case that
for each $i \in B$, $m_i = 1 + c_i/{\ga_i}$ and $(s, B \cup G \cup t)$ and $(s \cup B \cup G, t)$ must both
be min-cuts in the corresponding network (by the fact stated above). Using lemma \ref{lem.N} we get that
$\pp^{(0)}$ and $\mm^{(0)}$ are equilibrium prices and moneys for market $\cal M$. 
\end{proof}

Finally, by Theorem \ref{thm.reduce} we get:

\begin{corollary}
\label{cor.limit}
Algorithm \ref{alg.limit} converges to the Nash bargaining solution for \RNB.
\end{corollary}

\end{document}